\documentclass{lmcs}
\pdfoutput=1

\usepackage{lastpage}
\lmcsdoi{15}{2}{11}
\lmcsheading{}{\pageref{LastPage}}{}{}%
{May~15,~2018}{May~08,~2019}{}

\usepackage[utf8]{inputenc}

\usepackage{amssymb,xcolor,xspace}
\usepackage{amsmath,amsthm,graphicx,comment}
\usepackage{mathrsfs}
\usepackage{stmaryrd}
\usepackage{bbm}

\usepackage[normalem]{ulem}

\usepackage{hyperref}
\usepackage{tikz}
\usetikzlibrary{matrix,decorations.text,arrows,shapes,fit,shadows,calc,patterns,intersections}

\tikzstyle{non}=[inner sep=1pt]
\tikzstyle{lbox}=[rounded corners=5pt,draw=black!60,very thick,align=center]
\tikzstyle{wbox}=[rounded corners=5pt,align=center]
\tikzstyle{gbox}=[rounded corners=5pt,fill=green!20,align=center]
\tikzstyle{bbox}=[rounded corners=5pt,fill=blue!20,align=center]
\tikzstyle{rbox}=[rounded corners=5pt,fill=red!20,align=center]
\tikzstyle{obox}=[rounded corners=5pt,fill=yellow!35,align=center]
\tikzstyle{ledg}=[draw=black,very thick]
\tikzstyle{linc}=[fill=white,draw=black,rotate=90,inner sep=2pt,very thick,circle]

\newcommand{\nat}{\ensuremath{\mathbb{N}}\xspace}
\newcommand{\wpart}[2]{\ensuremath{[#1]_{#2}}\xspace}

\newcommand{\ppart}[1]{\wpart{#1}{\Pb}}
\newcommand{\qpart}[1]{\wpart{#1}{\Qb}}

\newcommand{\eqsu}[1]{\ensuremath{\sim_{#1}}\xspace}
\newcommand{\keqsu}{\eqsu{k}}

\newcommand{\sfa}{\ensuremath{\mathbb{S}_\eta}\xspace}
\newcommand{\sfl}[1]{\ensuremath{\lfloor #1\rfloor_\eta}\xspace}

\newcommand{\benrp}[2]{\ensuremath{{\llbracket #1\rrbracket_{#2}}}\xspace}
\newcommand{\benrd}[1]{\benrp{#1}{d}}
\newcommand{\benr}[1]{\ensuremath{{\llbracket #1 \rrbracket}}\xspace}

\newcommand{\fo}{\ensuremath{\textsf{FO}}\xspace}
\newcommand{\fow}{\mbox{\ensuremath{\fo(<)}}\xspace}
\newcommand{\fowm}{\mbox{\ensuremath{\fo(<,\MOD)}}\xspace}
\newcommand{\fows}{\mbox{\ensuremath{\fo(<,+1)}}\xspace}

\newcommand{\mso}{\ensuremath{\textsf{MSO}}\xspace}

\newcommand{\fod}{\ensuremath{\fo^2}\xspace}

\newcommand{\fodw}{\ensuremath{\fod(<)}\xspace}
\newcommand{\fodwm}{\ensuremath{\fod(<,\MOD)}\xspace}
\newcommand{\fodws}{\ensuremath{\fod(<,+1)}\xspace}
\newcommand{\fodwsm}{\ensuremath{\fod(<,+1,\MOD)}\xspace}

\newcommand{\sic}[1]{\ensuremath{\Sigma_{#1}}\xspace}
\newcommand{\siw}[1]{\ensuremath{\Sigma_{#1}(<)}\xspace}
\newcommand{\siwm}[1]{\ensuremath{\Sigma_{#1}(<,\MOD)}\xspace}
\newcommand{\siws}[1]{\ensuremath{\Sigma_{#1}(<,+1)}\xspace}
\newcommand{\siwsm}[1]{\ensuremath{\Sigma_{#1}(<,+1,\MOD)}\xspace}

\newcommand{\siwu}{\siw{1}}

\newcommand{\siwd}{\siw{2}}
\newcommand{\siwsd}{\siws{2}}

\newcommand{\siwt}{\siw{3}}
\newcommand{\siwst}{\siws{3}}

\newcommand{\siwq}{\siw{4}}

\newcommand{\pic}[1]{\ensuremath{\Pi_{#1}}\xspace}
\newcommand{\piw}[1]{\ensuremath{\Pi_{#1}(<)}\xspace}

\newcommand{\bsc}[1]{\ensuremath{\Bs\Sigma_{#1}}\xspace}
\newcommand{\bsw}[1]{\ensuremath{\Bs\Sigma_{#1}(<)}\xspace}
\newcommand{\bswm}[1]{\ensuremath{\Bs\Sigma_{#1}(<,\MOD)}\xspace}
\newcommand{\bsws}[1]{\ensuremath{\Bs\Sigma_{#1}(<,+1)}\xspace}
\newcommand{\bswsm}[1]{\ensuremath{\Bs\Sigma_{#1}(<,+1,\MOD)}\xspace}

\newcommand{\bswu}{\bsw{1}}

\newcommand{\bswd}{\bsw{2}}
\newcommand{\bswsd}{\bsws{2}}

\newcommand{\su}{\ensuremath{\textsf{SU}}\xspace}

\newcommand{\md}{\ensuremath{\textsf{MOD}}\xspace}

\newcommand{\vari}{quotienting Boolean algebra\xspace}

\newcommand{\pvari}{quotienting lattice\xspace}

\newcommand{\livari}{\textup{\bf li}-variety\xspace}

\newcommand{\plivari}{positive \textup{\bf li}-variety\xspace}

\newcommand{\plivaris}{positive \textup{\bf li}-varieties\xspace}

\newcommand{\varie}{variety\xspace}

\newcommand{\varies}{varieties\xspace}

\newcommand{\pvarie}{positive variety\xspace}

\newcommand{\pvaries}{positive varieties\xspace}

\newcommand{\eval}{\ensuremath{\text{\sc ev}}\xspace}

\newcommand{\Bs}{\ensuremath{\mathscr{B}}\xspace}
\newcommand{\Cs}{\ensuremath{\mathscr{C}}\xspace}
\newcommand{\Ds}{\ensuremath{\mathscr{D}}\xspace}

\newcommand{\Fs}{\ensuremath{\mathscr{F}}\xspace}

\newcommand{\Vs}{\ensuremath{\mathscr{V}}\xspace}

\newcommand{\Hb}{\ensuremath{\mathbf{H}}\xspace}

\newcommand{\Kb}{\ensuremath{\mathbf{K}}\xspace}
\newcommand{\Lb}{\ensuremath{\mathbf{L}}\xspace}
\newcommand{\Mb}{\ensuremath{\mathbf{M}}\xspace}

\newcommand{\Pb}{\ensuremath{\mathbf{P}}\xspace}
\newcommand{\Qb}{\ensuremath{\mathbf{Q}}\xspace}

\newcommand{\Ub}{\ensuremath{\mathbf{U}}\xspace}
\newcommand{\Vb}{\ensuremath{\mathbf{V}}\xspace}

\newcommand{\frS}{\ensuremath{\mathbbm{S}}\xspace}

\newcommand{\cmult}{\ensuremath{\mathbin{\scriptscriptstyle\bullet}}}

\theoremstyle{plain}
\newtheorem{theorem}{Theorem}[section]
\newtheorem{proposition}[theorem]{Proposition}
\newtheorem{lemma}[theorem]{Lemma}
\newtheorem{corollary}[theorem]{Corollary}
\newtheorem{fct}[theorem]{Fact}
\theoremstyle{definition}
\newtheorem{remark}[theorem]{Remark}
\newtheorem{example}[theorem]{Example}
\newtheorem{pro-example}[theorem]{Example}
\newtheorem{pro-remark}[theorem]{Remark}

\def\proofof#1{\begin{proof}[Proof of~#1]}
\def\eopo{\end{proof}}

\let\phi\varphi

\def\inv{^{-1}}
\let\epsilon\varepsilon

\def\lcm{\textsf{lcm}}

\newcommand{\MOD}{\ensuremath{\mathbf{MOD}}\xspace}

\def\D{\textbf{D}}

\renewcommand{\frS}{\ensuremath{\boldsymbol{S}}\xspace}

\title{Covering and separation for logical fragments with modular predicates}

\author[T. Place]{Thomas Place\rsuper{{a,d}}}	
\address{$\lsuper{a}$Univ. Bordeaux, CNRS,  Bordeaux INP, LaBRI, UMR 5800, F-33400, Talence, France}	
\thanks{Work supported by the ANR (project \textsc{DeLTA}, ANR-16-CE40-0007). \email{thomas.place@labri.fr, varun.ramanathan@u-bordeaux.fr, pascal.weil@labri.fr}}

\author[V. Ramanathan]{Varun Ramanathan\rsuper{{a,b,c}}}	
\address{\lsuper{b}Chennai Mathematical Institute, Chennai, India}

\author[P. Weil]{Pascal Weil\rsuper{{a,c}}}	
\address{\lsuper{{c}}CNRS, ReLaX, UMI 2000, Chennai, India}
\address{\lsuper{{d}}Institut Universitaire de France}

\begin{document}
\begin{abstract}
	For every class \Cs of word languages, one may associate a decision problem called \Cs-separation. Given two regular languages, it asks whether there exists a third language in \Cs containing the first language while being disjoint from the second one. Usually, finding an algorithm deciding \Cs-separation yields a deep insight on \Cs.
	
	We consider classes defined by fragments of first-order logic. Given such a fragment, one may often build a larger class by adding more predicates to its signature. In the paper, we investigate the operation of enriching signatures with \emph{modular predicates}. Our main theorem is a generic transfer result for this construction. Informally, we show that when a logical fragment is equipped with a signature containing the successor predicate, separation for the stronger logic enriched with modular predicates reduces to separation for the original logic. This result actually applies to a more general decision problem, called the covering problem. 
\end{abstract}
\maketitle

\section{Introduction}

\paragraph{\bf Context.} Logic is a powerful tool for the specification of the behavior of systems, which is classically modeled by formal languages of words, trees and other discrete structures. Monadic second-order (\mso), first-order (\fo) and their fragments play a prominent role in this context, especially in the case of languages of finite words (starting with the work of Büchi, see \cite{1994:Straubing} for a general overview). An important set of problems deals with characterizing the expressive power of particular logics. Given a logic \Fs, one would want a decision procedure for the \emph{\Fs-membership problem}: given a regular language $L$, decide whether $L$ is definable by a sentence of \Fs (i.e. whether $L$ belongs to the class of languages associated to \Fs). In practice, obtaining such an algorithm requires a deep understanding of \Fs: one needs to understand all the properties that can be expressed with \Fs.

Membership is known to be decidable for many fragments. The most celebrated result of this kind is the decidability of \fow-membership (first-order logic equipped the linear ordering) which is due to Schützenberger, McNaughton and Papert~\cite{sfo,mnpfo}. Another well-known example is the solution of Thérien and Wilke~\cite{1998:TherienWilke} for the two-variable fragment of first-order logic (\fodw). However, membership remains open for several natural fragments. A famous open question is to obtain membership algorithms for every level in the quantifier alternation hierarchy of first-order logic, which classifies it into levels \siw{n} and \bsw{n}. This problem has been investigated for years, starting with the theorem of Simon~\cite{simon75} which states that \bsw{1}-membership is decidable. However, despite all these efforts, progress has been slow and decidability is only known for $n \leq 4$ in the case of \siw{n} and for $n \leq 2$ in the case of $\bsw{n}$~\cite{pwdelta,pzqalt,pseps3,pseps3j} (we refer the reader to~\cite{jep-dd45,2015:PlaceZeitoun-SIGLOG,2017:PlaceZeitoun-CSR} for detailed surveys on the problem). A key point is that the latest results on this question (for \siwt, \bswd and \siwq) involve considering new decision problems generalizing membership: \emph{separation} and \emph{covering}.

For a fixed fragment \Fs, the \emph{\Fs-separation problem} takes \emph{two} input regular languages and asks whether there exists a third language, definable in \Fs, containing the first language and disjoint from the second one. Covering~\cite{2018:PlaceZeitoun,2016:PlaceZeitoun} is even more general: it takes two different objects as input: a regular language $L$ and a finite set of regular languages \Lb. It asks whether there exists an \Fs-cover \Kb of $L$ (\emph{i.e.}, a finite set of languages definable in \Fs whose union contains $L$) such that no language $K \in \Kb$ meets every language in \Lb. Separation corresponds to the special case of covering when the set \Lb is a singleton. Membership is the particular case of separation where the input languages are of the form $L$ and $A^*\setminus L$. The investigation of these two problems has been particularly fruitful. While obtaining an algorithm for separation or covering is usually more difficult than for membership, it is also more rewarding with respect to the insight one gets on the logical fragment investigated. Furthermore, solutions for these two problems are often more robust than those for membership. A striking example is a transfer theorem of~\cite{pzqalt} which can be applied to the quantifier alternation hierarchy: for every $n \geq 1$, $\siw{n+1}$-membership can be effectively reduced to \siw{n}-separation.

Separation is known to be decidable, for instance, for $\bsw{1}$ (Place, van Rooijen, Zeitoun \cite{pvzmfcs13} and, independently, Czerwinski, Martens and Masopust \cite{2013:CzerwinskiMartensMasopust}). The problem of separation was actually studied earlier, under a different guise: Almeida \cite{1999:Almeida} showed, if \Vs is a variety of languages and \Vb is the corresponding pseudovariety of monoids (in Eilenberg's correspondence, see \cite{1986:Pin}), that the \Vs-covering problem and the problem of computing \emph{\Vb-pointlikes} (a notion with deep roots in algebra and topology which we will not need to discuss here) reduce to each other in a natural fashion. It follows that earlier work of Henckell \cite{1988:Henckell} (see also \cite{2010:HenckellRhodesSteinberg-b,2017:GoolSteinberg,2018:GoolSteinberg}) on aperiodic pointlikes shows that \fow-separation and covering are decidable. For a a fully language-theoretic proof of this result, see Place and Zeitoun \cite{pzfoj}.

In the paper, we investigate the separation and covering problems for several families of logical fragments, built from weaker fragments using a generic operation: \emph{enrichment by modular predicates}. This operation was first introduced by Barrington \emph{et al.} \cite{1992:BarringtonComptonStraubing} in their study of circuit complexity. It is a natural extension also in the following sense: the first standard examples of regular languages that are not in \fow are those that involve modulo counting (\emph{e.g.}, the set of words of even length). \emph{Modular predicates} introduce just that capability: for every natural number $d \geq 1$ and every $i<d$, a unary predicate is available, which selects the positions in a word that are congruent to $i$ modulo $d$. Adding these predicates to a fragment of \fow allows the concise specification of languages that are not \fow-definable (it strictly increases the expressive power of the logic), while remaining within \mso-definability, and hence within the realm of regular languages.

Before we give an overview of the state of the art, let us comment on our method, or rather on the point of view we adopted in this paper. Specialists of the domain are well aware that this sort of problems, linking automata theory and logical specifications, can be approached from different angles. One of those is to rely heavily on algebraic models (following, say, \cite{1976:Eilenberg,1986:Pin,1994:Straubing}). Here we choose instead to remain as close as possible to purely language-theoretic arguments as in, \emph{e.g.}, \cite{2017:PlaceZeitoun-CSR,2018:PlaceZeitoun}. This way, we gain a little generality and we avoid introducing sophisticated machinery. One side-effect is that we are led to translating certain ideas that were already present in the framework of logic or algebra, under the restrictions imposed by that framework (and we explicitly mention these situations). Our main objective in opting for this approach is to, hopefully, make the paper more accessible to a wider audience.

\paragraph{\bf State of the art.} Currently, no results are known for the separation and covering problems associated to fragments enriched with modular predicates. However, earlier results considered the status of membership when a fragment of \fow is enriched with modular predicates: Barrington et al. showed that \fowm-membership is decidable~\cite{1992:BarringtonComptonStraubing} (see also \cite{1994:Straubing}). Chaubard et al.~\cite{2006:ChaubardPinStraubing-LICS} show the same result for the extensions of \siwu and \bswu: \siwm{1} and \bswm{1}. Kufleitner and Walter \cite{2015:KufleitnerWalter} show the decidability of membership for \siwm{2}. Finally, Dartois and Paperman show that the extension of \fodw with modular predicates (\fodwm) has decidable membership as well~\cite{2013:DartoisPaperman}. In \cite{2015:DartoisPaperman}, they extend their result to a wider range of logical fragments enriched with modular predicates (see also the unpublished \cite{2015:DartoisPaperman-unpub}).

Unfortunately, each of these decidability results is proved using a specific argument which deals directly with a particular fragment enriched with modular predicates. There are many fragments of first-order logic. Consequently, it would be preferable to avoid an approach which considers each of them independently. Instead, we would like to have a \emph{generic transfer theorem} which, given a fragment \Fs, lifts a solution of the decision problems for \Fs to the stronger variant of \Fs enriched with modular predicates.

It turns out that such a transfer theorem is already known for another natural operation on fragments of first-order logic: enrichment by the local predicates (the successor binary relation ``$+1$'' and the \textsf{min} and \textsf{max} unary predicates). In~\cite{2015:PlaceZeitoun,2017:PlaceZeitoun}, Place and Zeitoun showed that separation and covering are decidable for the following logical fragments enriched with the local predicates: \fodws, \siws{1}, \bsws{1}, \siwsd, \bswsd and \siwst (the enrichment of \fow is omitted as its expressive power is not increased). A key point is that all these results rely on the same generic reduction which is obtained by translating the problem to a language theoretic statement. Translating ideas from Straubing's seminal paper~\cite{1985:Straubing} on algebraic methods to solve the membership problem when local predicates are added to a fragment into a purely language-theoretic setting, Place and Zeitoun investigate an operation on classes of languages, namely $\Cs \mapsto \Cs\circ\su$ (see Section~\ref{sec:enrich} for definitions). It satisfies the two following properties:
\begin{enumerate}
	\item In most cases, if \Fs is a fragment of first-order logic, the enrichment of \Fs with local predicates corresponds to the class $\Fs \circ \su$. In~\cite{2018:PlaceZeitoun}, this is proved for two-variable first-order logic (\fodw) and levels in the alternation hierarchy (\siw{n} and \bsw{n}).
	\item The main theorem of~\cite{2018:PlaceZeitoun} states that the  separation and covering problems are decidable on $\Cs\circ\su$ if they are on \Cs.	This result highlights the robustness of separation and covering: such a theorem fails for membership.	
\end{enumerate}
The decidability results for the extended fragments that we mentioned above are then obtained from already established results for the separation and covering problems associated to \fodw, \siwu, \bswu, \siwd, \bswd and \siwt (see~\cite{2014:PlaceZeitoun,2013:CzerwinskiMartensMasopust,2018:PlaceZeitoun,pseps3j}).

\paragraph{\bf Contribution.} Our results apply to enrichment by modular predicates and follow the same scheme, with an interesting twist. We investigate an operation on classes of languages whose origins can be found implicitly or explicitly in \cite{1992:BarringtonComptonStraubing,2006:ChaubardPinStraubing-LICS}, written $\Cs \mapsto \Cs\circ\md$ (again, see Section~\ref{sec:enrich} for definitions). Then we show two properties which are similar to those of $\Cs \circ \su$.
\begin{enumerate}
	\item If \Cs is the class corresponding to a fragment \Fs of first-order logic satisfying some mild hypotheses, then $\Cs \circ \md$ corresponds to the extension of \Fs by modular predicates. Let us point out that this result is generic, which makes it stronger, and much simpler to establish than the corresponding result for $\Cs \circ \su$ (which requires a separate proof for each fragment).
	\item We show that, for every class \Cs of the form $\Cs = \Ds \circ \su$, the separation and covering problems are decidable on $\Cs \circ \md$ if they are on \Cs. This result is weaker than that for $\Cs \circ \su$: it does not apply to all classes, only those of the form $\Ds \circ \su$. 
\end{enumerate}
Consequently, our results complement those of Place and Zeitoun~\cite{pzfoj,2017:PlaceZeitoun} in a natural way. Combining the two theorems yield that for every class \Cs, the  separation and covering problems are decidable on $(\Cs\circ\su) \circ \md$ if they are on \Cs.

Going back to the logical point of view, this shows that for most fragments \Fs, enriching \Fs with the local and modular predicates simultaneously preserves the decidability of separation and covering. In the paper, we use this result to show that separation and covering are decidable for several fragments: \fowm, \fodwsm, \siwsm{n} for $n =1,2,3$ and \bswsm{n} for $n =1,2$.

\paragraph{\bf Plan of the paper}
The paper is organized as follows. Classes of languages and the membership, separation and covering problems are introduced in Section~\ref{sec:prelims}. Fragments of first-order logic, and the notion of enrichment of a fragment by the adjunction of new predicates (typically: the local and the modular predicates) is introduced in Section~\ref{sec:logic}. The enrichment operation on classes of languages, namely the operation $\Cs,\Ds\mapsto \Cs\circ\Ds$, is studied in Section~\ref{sec:enrich}. We also discuss in that section the language-theoretic impact of enriching a logical fragment with local and, especially, modular predicates. We formulate our main theorem in Section~\ref{sec: main theorem}: it reduces the covering and separation problems in a class of the form $(\Cs\circ\su)\circ\md$ to the same problem in $\Cs\circ\su$.

The proof of Theorem~\ref{thm:main} is given in two parts: in Section~\ref{sec:block}, we translate the problem to another language theoretic problem (in terms of the operation of block abstraction, introduced here), which turns out to be easier to handle. The proof of the theorem itself is given in Section~\ref{sec:transfer}. 

\section{Preliminaries}
\label{sec:prelims}

In this section, we introduce classes of languages and the decision problems that we shall consider.

\subsection{Classes of languages}

An \emph{alphabet} is a finite non-empty set of symbols that we call \emph{letters}. Given an alphabet $A$, we denote by $A^*$ the set of all finite sequences (known as \emph{words}) of letters in $A$, including the empty word denoted by $\varepsilon$. Given a word $w \in A^*$ we write $|w| \in \nat$ for the length of $w$. The set of non-empty words is written $A^+$ ($A^+ = A^* \setminus \{\varepsilon\}$). A subset of $A^*$ is called a \emph{language}.

A \emph{class of languages} \Cs is a correspondence $A \mapsto \Cs(A)$, defined for all alphabets $A$, where $\Cs(A)$ is a set of languages over $A$. All classes considered in the paper are included in the class of regular languages, this will be implicitly assumed from now on. These are the languages which are accepted by a nondeterministic finite automaton (see e.g. \cite{1997:Kozen}). Here, we shall work with the classical monoid-theoretic characterizations, which we now recall.

Recall that a \emph{semigroup} $S$ is a set equipped with an associative binary operation (called the \emph{product}) and a \emph{monoid} is a semigroup $M$ that has an identity element (which we denote by $1_M$). If $A$ is an alphabet, then $A^*$ is a monoid under word concatenation and $A^+$ is a semigroup under the same operation.

A morphism $\eta\colon A^* \to M$, into a finite monoid, is said to recognize a language $L \subseteq A^*$ if $L = \eta\inv(X)$ for some subset $X$ of $M$, that is, if $L = \eta\inv(\eta(L))$. It is well-known that a language is regular if and only if it is recognized by a morphism into a finite monoid \cite{1986:Pin}.

%

\subsection{Closure properties}

In the paper, we only work with classes of languages satisfying robust closure properties that we describe below.

\paragraph{Boolean operations.} A class \Cs of languages is a \emph{lattice} when $\Cs(A)$ is closed under finite unions and intersections for every alphabet $A$: that is, $\emptyset,A^* \in \Cs(A)$ and for every $L_1,L_2 \in \Cs(A)$, we have $L_1 \cup L_2 \in \Cs(A)$ and $L_1 \cap L_2 \in \Cs(A)$. A \emph{Boolean algebra} is a lattice \Cs that is additionally closed under complement: for every alphabet $A$ and every $L \in \Cs(A)$, we have $A^* \setminus L \in \Cs(A)$.

\paragraph{Quotients.}  Recall that if $L$ is a language in $A^*$ and $u\in A^*$, then the \emph{left} and \emph{right quotients} $u\inv L$ and $Lu\inv$ are defined as follows
$$u\inv L = \{v \in A^* \mid uv\in L\}\quad\textrm{ and }\quad L u\inv = \{v \in A^* \mid vu\in L\}.$$
We say that a class \Cs is \emph{quotient-closed} when it is closed under taking (left and right) quotients by words of~$A^*$.

\paragraph{Inverse morphisms.} A class \Cs is \emph{closed under inverse morphisms} when, for all alphabets $A,B$,  morphism $\alpha\colon A^* \to B^*$ and language $L \in \Cs(B)$, we have $\alpha\inv(L) \in \Cs(A)$.

We shall also consider a weaker variant of this closure property, namely \emph{closure under inverse length increasing morphisms}. A morphism $\alpha$ as above is \emph{length increasing} when $|\alpha(w)| \geq |w|$ for every $w \in A^*$. A class \Cs is closed under inverse length increasing morphisms when for every such morphism and every language $L \in \Cs(B)$, we have $\alpha\inv(L) \in \Cs(A)$. Observe that by definition, any class \Cs which is closed under inverse morphisms is also closed under inverse length increasing morphisms.

\begin{remark}
	The morphism $\alpha$ is length increasing if and only if $\alpha(a) \neq \varepsilon$ for any $a \in A$. Such morphisms are also called \emph{non-erasing}.
\end{remark}

\paragraph{Varieties.} We call \emph{\pvarie} (resp. \emph{\varie}) a \pvari (resp. \vari) which is closed under inverse morphisms. Similarly, we call \emph{\plivari} (resp \emph{\livari}) a \pvari (resp. \vari) which is closed under inverse length increasing morphisms.

\subsection{Decision problems}

Our main objective in this paper is to investigate two decision problems: separation and covering.  Both problems are parametrized by an arbitrary class of languages \Cs. We start with the definition of separation.

Given two languages $L_1,L_2$ over some alphabet $A$, we say that a language $K \subseteq A^*$ \emph{separates $L_1$ from $L_2$} if $L_1 \subseteq K$ and $L_2 \cap K = \emptyset$. Furthermore, given a class \Cs, we say that $L_1$ is \emph{\Cs-separable from} $L_2$ if there exists $K \in \Cs(A)$ which separates $L_1$ from $L_2$.

Given a class of languages \Cs, the \emph{\Cs-separation problem} takes as input two regular languages $L_1$ and $L_2$ and asks whether $L_1$ is \Cs-separable from $L_2$.

\begin{remark}
	Note that when \Cs is not closed under complement, this definition is not symmetric: it could happen that $L_1$ is \Cs-separable from $L_2$ while $L_2$ is not \Cs-separable from $L_1$.
\end{remark}

The covering problem is a generalization of separation that was originally defined in~\cite{2018:PlaceZeitoun,2016:PlaceZeitoun}. As pointed out in the introduction, in the particular case of varieties of regular languages, the covering problem is a translation of the more algebraic problem of the computation of pointlikes \cite{1999:Almeida}.

Given a language $L$ in $A^*$, a \emph{cover of $L$} is a finite set of languages \Kb in $A^*$ such that $L \subseteq \bigcup_{K \in \Kb} K$. For a given class \Cs, a \emph{\Cs-cover} of $L$ is  a cover \Kb of $L$ such that every $K \in \Kb$ belongs to \Cs. Additionally, given a finite multiset\footnote{We speak of multiset here in order to allow \Lb to have several copies of the same language. This is natural as the input for the covering problem is given by monoid morphisms recognizing the languages in \Lb, and it may happen that several distinct recognizers define the same language. This is harmless in practice.} of languages \Lb, we say that a finite set of languages \Kb is \emph{separating for \Lb} if for any $K\in\Kb$, there exists $L \in \Lb$ such that $K \cap L = \emptyset$ (that is: no language in \Kb intersects every language in \Lb).

Now consider a class \Cs. Given a language $L_1$ and a finite multiset of languages $\Lb_2$, we say that the pair $(L_1,\Lb_2)$ is \emph{\Cs-coverable} if there exists a \Cs-cover of $L_1$ which is separating for $\Lb_2$. The \emph{\Cs-covering problem} takes as input a regular language $L_1$ and a finite multiset of regular languages $\Lb_2$, and it asks whether $(L_1,\Lb_2)$ is $\Cs$-coverable.

We complete this definition by showing why covering generalizes separation, at least when the class \Cs is a lattice.

\begin{fct} \label{fct:septocove}
	Let \Cs be a lattice and $L_1,L_2$ two languages. Then $L_1$ is \Cs-separable from $L_2$, if and only if $(L_1,\{L_2\})$ is \Cs-coverable.
\end{fct}

\begin{proof}
	Assume that $L_1$ is \Cs-separable from $L_2$ and let $K\in\Cs$ be a separator. It is immediate that $\Kb = \{K\}$ is a separating \Cs-cover for $(L_1,\{L_2\})$. Conversely, assume that $(L_1,\{L_2\})$ is \Cs-coverable and let \Kb be the separating \Cs-cover. The union $K$ of the languages in \Kb is in \Cs since \Cs is a lattice, and $L_1 \subseteq K$ by definition of a cover. Moreover, no language in \Kb intersects $L_2$ since \Kb was separating for $\{L_2\}$, that is, $L_2 \cap K = \emptyset$. Therefore $K \in \Cs$ separates $L_1$ from $L_2$. 
\end{proof}



\section{Logic and modular predicates}
\label{sec:logic}

We introduce here the formal definition of fragments of first-order logic and the classes of languages they define. 

\subsection{First-order logic}

Let us fix an alphabet $A$. Any word $w \in A^*$ may be viewed as a logical structure whose domain is the set of positions $\{0,\dots,|w|-1\}$ in $w$. In first-order logic, one can quantify over these positions and use a pre-determined signature of predicates.

More precisely: a \emph{predicate} is a symbol $P$ together with an integer $k \in \nat$ called the \emph{arity} of $P$. Moreover, whenever we consider a predicate $P$ of arity $k$, we shall assume that an \emph{interpretation} of $P$ is fixed: for every word $w \in A^*$, $P$ is interpreted as a $k$-ary relation over the set of positions of $w$. A \emph{signature} (over the alphabet $A$) is a (possibly infinite) set of predicate symbols $\frS = \{P_1,\dots,P_\ell,\dots\}$, all interpreted over words in $A^*$. Given such a signature \frS, we write $\fo[\frS]$ for the set of all first-order formulas over~\frS. As usual, we call sentence a formula with no free variable.

We use the classical semantic of first-order formulas to interpret formulas of $\fo[\frS]$ on a word $w \in A^*$. In particular, every $\fo[\frS]$ sentence $\varphi$ defines a language $L \subseteq A^*$. It contains all words $w \in A^*$ satisfying $\varphi$: $L = \{w \in A^* \mid w \models \varphi\}$. 

\begin{example}
	Assume that $A = \{a,b\}$ and consider the signature $\frS = \{a,b,<\}$ where $a$ (resp. $b$) is a unary predicate which selects positions labeled with the letter $a$ (resp. $b$) and $<$ is a binary predicate interpreted as the (strict) linear order. Then the following is an $\fo[\frS]$ sentence:
	\[
	\exists x \exists y\  x < y \wedge a(x) \wedge b(x) \wedge \neg (\exists z \ x < z \wedge z < y). 
	\]
	It defines the language $A^*abA^*$. 
\end{example}

In the paper, we shall consider signatures containing specific predicates. We present them now.

\paragraph{Label predicates.} All the signatures that we consider include the label predicates. For each letter $a\in A$, the unary \emph{label predicate} $a(x)$ is interpreted as the unary relation selecting all positions whose label is $a$. Observe that these predicates depend on the alphabet $A$ that we are using.  Abusing notation, we write $A$ for the set of label predicates over alphabet $A$.

\paragraph{Linear order.} Another predicate that is also always contained in our signatures is the binary predicate $<$ which is interpreted as the (strict) linear order over the positions.

Moreover, two natural sets of predicates are of particular interest for the paper: the local predicates and the modular predicates.

\paragraph{Local predicates.} There are four local predicates. They are as follows:
\begin{itemize}
	\item the binary predicate  $+1$ interpreted as the successor relation between positions;
	\item the unary predicate \textsf{min} selecting the leftmost position in the word;
	\item the unary predicate \textsf{max} selecting the rightmost position in the word;
	\item the constant (0-ary) predicate $\varepsilon$ which holds exactly when the word is empty.
\end{itemize}

\paragraph{Modular predicates.} There are infinitely many modular predicates. For every natural number $d > 1$ and $0 \le i < d$, 
\begin{itemize}
	\item the unary predicate $\MOD^d_i$ selects the positions $x$ that are congruent to $i$ modulo $d$;
	\item the constant $\D^d_i$ holds when the length of the word is congruent to $i$ modulo~$d$.
\end{itemize}
We denote by \MOD\ the (infinite) set of all modular predicates.

These are all the predicates that we shall consider. An important observation is that the order predicate $<$, the local and the modular predicates are examples of \emph{structural} predicates (also known as \emph{numerical} predicates, see \cite{1994:Straubing}): a $k$-ary predicate $P$ is structural if its interpretation is defined for every alphabet and is independent from the labels. More precisely, given two words $w,w'$ (possibly over different alphabets) having the same length $\ell$ and $k$ positions $i_1,\dots,i_k \leq \ell$, $P$ satisfies the following property:
\[
\text{$P(i_1,\dots,i_k)$ holds in $w$} \quad \text{if and only if} \quad \text{$P(i_1,\dots,i_k)$ holds in $w'$.}
\]
We say that a signature \frS is \emph{structural} if it contains only structural predicates.

\subsection{Fragments of first-order logic}\label{sec: fragments and logical enrichment}

A \emph{fragment} $\Fs$ of first-order logic consists in the specification of a (possibly finite) set $V_\Fs$ of variables and a correspondence $\Fs\colon \frS \mapsto \Fs[\frS]$ which associates with every signature \frS a set $\Fs[\frS] \subseteq \fo[\frS]$ of formulas over the signature \frS using only the variables in $V_\Fs$ and which satisfies the following properties:
\begin{itemize}
	\item for every signature \frS, every quantifier-free $\fo[\frS]$-formula belongs to $\Fs[\frS]$;
	\item $\Fs[\frS]$ is closed under conjunction and disjunction;
	\item $\Fs$ is closed under quantifier-free substitution: if $\frS,\frS'$ are signatures and $\phi$ is a formula in $\Fs[\frS]$, then $\Fs[\frS']$ contains every formula $\varphi'$ obtained by replacing each atomic formula in $\varphi$ by a quantifier-free formula of $\Fs[\frS']$.
\end{itemize}

\begin{example}
\fo itself is a fragment of first-order logic, and so is $\fo^2$, where $\fo^2[\frS]$ is the set of formulas in $\fo[\frS]$ which use at most two variable names.

The classes in the quantifier alternation hierarchy also provide interesting examples: given $n \geq 1$, a formula of $\fo[\frS]$ is in $\sic{n}[\frS]$ (resp. $\pic{n}[\frS]$) if its prenex normal form has $n-1$ quantifier alternations (i.e. $n$ blocks of quantifiers) and starts with an existential (resp. universal) quantifier, or if it has at most $n-2$ quantifier alternations. For example, a sentence whose prenex normal form is
\[
\exists x_1\enspace \forall x_2\enspace  \forall x_3 \enspace\exists x_4 \enspace\forall x_5
\ \varphi(x_1,x_2,x_3,x_4,x_5) \quad \text{(with $\varphi$ quantifier-free)}
\]
\noindent
is {\sic 4}. Clearly, $\sic{n+1}[\frS]$ contains both $\sic{n}[\frS]$ and $\pic{n}[\frS]$. Moreover, \sic{n} and \pic{n} formulas are not closed under negation and it is standard to also consider \bsc{n}, namely the Boolean combinations of \sic{n} formulas. 
\end{example}

A fragment \Fs and a structural signature \frS determine a class of languages, denoted by $\Fs(\frS)$. For every alphabet $A$, $\Fs(\frS)(A)$ contains the languages $L \subseteq A^*$ which can be defined by a sentence of $\Fs[A,\frS]$. That is, a sentence of \Fs over the signature containing the label predicates over $A$ and the structural predicates in \frS (whose interpretation is defined independently of the alphabet).

It is well known that the classes $\fow$, $\fodw$, $\bsw{n}$ are \varies, and that the classes $\siw{n}$ and $\piw{n}$ are \pvaries. The latter form a strict hierarchy whose union is \fow~\cite{1978:BrzozowskiKnast}.

\paragraph{Logical enrichment.} In the paper, we are not interested in a particular class. Instead, we consider a generic operation that one may apply to logically defined classes. Let \Fs be a fragment, \frS a structural signature and consider the class $\Fs(\frS)$. One may define a larger class by enriching \frS with additional predicates. We are mainly interested in enrichment by \emph{modular predicates}. If \Fs is a fragment and \frS is a structural signature, we write $\Fs(\frS,\MOD)$ for the class associated to \Fs and the enriched structural signature $\frS \cup \MOD$.

Our main objective is to investigate the separation and covering problems for classes of the form $\Fs(\frS,\MOD)$. Ideally, we would like to have a transfer theorem: a generic effective reduction from $\Fs(\frS,\MOD)$-covering to $\Fs(\frS)$-covering. Getting such a result remains an open problem. We obtain a slightly weaker one: while we do present such a reduction, it is only correct when the smaller class $\Fs(\frS)$ satisfies specific hypotheses. Informally, we require the original signature \frS to contain the set of local predicates ($+1$, \textsf{min}, \textsf{max} and $\varepsilon$).

This motivates us to introduce a notation for enrichment by local predicates. If \Fs is a fragment of first-order logic and \frS a structural signature, we denote by $\Fs(\frS,+1)$ the class of languages associated to \Fs and the enriched structural signature $\frS \cup \{+1,\textsf{min},\textsf{max}, \varepsilon\}$.

\begin{remark} \label{rem:fosucc}
	By combining the two above operations, every fragment \Fs and structural signature \frS yield four classes: $\Fs(\frS)$, $\Fs(\frS,+1)$, $\Fs(\frS,\MOD)$ or $\Fs(\frS,+1,\MOD)$. These four classes are distinct in most cases. However, it may not be the case in specific situations. We shall encounter one important such situation in the paper: \fow and \fows coincide (as well as \fowm and $\fo(<,+1,\MOD)$). This is because the local predicates are easily defined from the linear order when quantifications are unrestricted. For example, $x + 1= y$ is expressed by the formula $x < y \wedge \neg\exists z \left(x < z \wedge z< y \right)$. 
\end{remark}

It was shown in~\cite{2015:PlaceZeitoun,2017:PlaceZeitoun} that under mild assumptions on $\Fs(\frS)$ (general enough to capture all relevant examples), $\Fs(\frS,+1)$-covering can be effectively reduced to $\Fs(\frS)$-covering. Our main theorem (presented in Section~\ref{sec: main theorem}) extends this result: under the same assumptions (on $\Fs(\frS)$), there is another effective reduction from $\Fs(\frS,+1,\MOD)$-covering to $\Fs(\frS,+1)$-covering (and therefore to $\Fs(\frS)$-covering as well by transitivity). Combined with previously known results, this yields the decidability of covering for several classes of the form $\Fs(\frS,+1,\MOD)$.

Our approach is similar to that taken in~\cite{2017:PlaceZeitoun}. We exploit the fact that the addition of local or modular predicates to a structural signature admits a nice language-theoretic characterization. This is discussed in Section~\ref{sec:enrich}. We then formulate our main theorem precisely in Section~\ref{sec: main theorem}.


\section{The enrichment operation for classes of languages}
\label{sec:enrich}
Let $\Cs$ and $\Ds$ be classes of languages. We define in Section~\ref{sec: enrichment} the \emph{\Ds-enrichment of \Cs}, written $\Cs \circ \Ds$. Before we get to the precise definition, let us formulate two remarks.

\begin{remark}
	Enrichment is the language theoretic counterpart of an algebraic operation defined between varieties of semigroups: the \emph{wreath product}. In fact, we use the same notation as for the wreath product and our definition (taken from~\cite{2017:PlaceZeitoun}) is based on Straubing's so-called wreath product principle~\cite{1985:Straubing}. Much of the literature on this topic is written from an algebraic point of view, including the connection with logic. The language theoretic approach adopted here allows us to bypass some general machinery which is not needed in this context.
\end{remark}

\begin{remark}\label{rk: enrich with SU, MOD}
Whereas we define enrichment for arbitrary classes \Cs and \Ds, we will mainly be concerned with \Ds-enrichment when \Ds is the class \su of suffix languages or the class \md of modulo languages (defined below). If \Cs is defined by a fragment of first-order logic and a structural signature, these language-theoretic operations correspond to adding local or modular predicates to the signature, see Section~\ref{sec: enrichment and logic}. The main result of the paper, Theorem~\ref{thm:main}, states that, for any \pvarie \Cs, the covering problem associated to the class $(\Cs \circ \su) \circ \md$ reduces to the same problem for the class $\Cs \circ \su$. 
\end{remark}

\subsection{Enrichment of a class of languages}\label{sec: enrichment}

We first define a technical notion of \Pb-tagging. Let $A$ be an alphabet and \Pb a finite partition of $A^*$. For each word $u \in A^*$, we denote by $\ppart{u} \in \Pb$ the unique language in \Pb which contains $u$. We use the finite set $\Pb \times A$ as an extended alphabet and define a canonical map $\tau_\Pb\colon A^* \to (\Pb \times A)^*$ as follows. 

Let $w \in A^*$. If $w = \varepsilon$, then we let $\tau_\Pb(\varepsilon) = \varepsilon$. Otherwise, $w = a_1 \cdots a_n$ with $n \geq 1$ and $a_1,\dots,a_n \in A$, and we let $\tau_\Pb(w) = b_1 \cdots b_n$ with
\[
b_1 = (\ppart{\varepsilon},a_1)\in \Pb \times A \quad \text{and} \quad b_i = (\ppart{a_1 \cdots a_{i-1}},a_i) \in \Pb \times A \quad \text{for $2 \leq i \leq n$.}
\]
We call $\tau_\Pb(w)$ the \emph{\Pb-tagging} of $w$. It can be viewed as a simple relabeling: each position $i$ in $w$ is given a new label encoding its original label in $A$ and the unique language in \Pb which contains the prefix ending at position $i-1$.

\begin{example} \label{ex:tag}
	On alphabet $A = \{a,b\}$, consider the languages $P_0,P_1,P_2$, defined by $P_m = \{w \in A^* \mid |w| = m \mod 3\}$. Clearly, $\Pb = \{P_0,P_1,P_2\}$ is a partition of $A^*$. The \Pb-tagging of $w = babbbaaa$ is
	$\tau_\Pb(w) = (P_0,b)(P_1,a)(P_2,b)(P_0,b)(P_1,b)(P_2,a)(P_0,a)(P_1,a)$.
\end{example}  

\begin{remark}
	Note that the map $w \mapsto \tau_\Pb(w)$ is not a morphism, nor is it surjective in general. There are usually compatibility constraints between consecutive positions in $\tau_\Pb(w)$. This can be observed in Example~\ref{ex:tag}. On the other hand, $\tau_\Pb$ is clearly injective.
\end{remark}

The following fact is an immediate consequence of the definition.

\begin{fct} \label{fct:tagdecomp}
	Let $A$ be an alphabet and \Pb a finite partition of $A^*$. Then for any $a \in A$ and $u \in A^*$, we have $\tau_\Pb(ua) = \tau_\Pb(u) \cdot (\ppart{u},a)$.
\end{fct}

We now explain how one may use taggings to build new languages from those contained in a fixed class. Consider an arbitrary  class \Cs and an alphabet $A$. Let \Pb be a finite partition of $A^*$.  We say that a language $L \subseteq A^*$ is \emph{\Pb-liftable from \Cs} if there exist languages $L_P \in \Cs(\Pb \times A)$ for every $P \in \Pb$ such that
\[
L = \bigcup_{P \in \Pb} \left(\tau_\Pb\inv(L_P) \cap P\right).
\]
We may now define enrichment. Consider two classes of languages \Cs and \Ds. If $A$ is an alphabet, a \emph{finite} partition of $A^*$ into languages of \Ds is called a \emph{\Ds-partition of $A^*$}. The \emph{\Ds-enrichment of \Cs}, written $\Cs\circ \Ds$, is the class such that for every alphabet $A$, $(\Cs \circ \Ds)(A)$ consists of the languages $L \subseteq A^*$ which are \Pb-liftable from \Cs, for some \Ds-partition \Pb of $A^*$.

Before we investigate the properties of classes that are built with enrichment, let us present two technical lemmas. We shall use them to select appropriate \Ds-partitions for building languages in $\Cs \circ \Ds$.

\begin{lemma} \label{lem:refine}
	Let \Cs be a class closed under inverse length increasing morphism and $A$ an alphabet. Let \Pb and \Qb be finite partitions of $A^*$ such that \Pb refines \Qb. Then every language $L \subseteq A^*$ which is \Qb-liftable from \Cs is also \Pb-liftable from \Cs. 
\end{lemma}

\begin{proof}
	Let $L \subseteq A^*$ be \Qb-liftable from \Cs. There are languages $L_Q \in \Cs(\Qb \times A)$ for every $Q \in \Qb$ such that
	\[
	L = \bigcup_{Q \in \Qb} \left(\tau_\Qb\inv(L_Q) \cap Q\right).
	\]
	Since \Pb refines \Qb, for every $P \in \Pb$, there exists a unique language $Q_P \in \Qb$ such that $P \subseteq Q_P$. We use this to define a length increasing morphism $\alpha: (\Pb \times A)^* \to (\Qb \times A)^*$. For every letter $(P,a) \in \Pb \times A$, we define $\alpha((P,a)) = (Q_P,a)$. It is now simple to verify from the definitions that
	\[
	L = \bigcup_{P \in \Pb} \left(\tau_\Pb\inv(\alpha\inv(L_{Q_P})) \cap P\right).
	\]
	Since \Cs is closed under inverse length increasing morphism, it is immediate that for every $P \in \Pb$, we have $\alpha\inv(L_{Q_P}) \in \Cs(\Pb \times A)$. Therefore $L$ is \Pb-liftable from \Cs, as desired.	
\end{proof}

The second lemma shows that we may always work with a \Ds-partition which satisfies some additional properties (provided that \Ds is sufficiently robust). Recall that a finite partition \Pb of $A^*$ is a (finite index) \emph{congruence} if $\ppart{u}\ppart{v}$ is contained in $\ppart{uv}$ for every $u, v \in A^*$. In that case, the set \Pb is a monoid for the operation $\ppart{u} \cmult \ppart{v} = \ppart{uv}$ and the map $w \mapsto \ppart{w}$ is a morphism from $A^*$ to \Pb. We talk of a \emph{\Ds-congruence} if \Pb is a \Ds-partition and a finite index congruence.

\begin{lemma} \label{lem:partmult}
	Let \Cs be a class closed under inverse length increasing morphism and \Ds a \vari. Let $A$ be an alphabet and $L \in (\Cs \circ \Ds)(A)$. Then there exists a \Ds-congruence \Pb of $A^*$ such that $L$ is \Pb-liftable from \Cs.
\end{lemma}

\begin{proof}
	By definition $L$ is \Qb-liftable from an arbitrary \Ds-partition \Qb of $A^*$. Thus, by Lemma~\ref{lem:refine}, it suffices to show that there exists a \Ds-congruence \Pb which refines \Qb. Since \Qb is a finite set, it is folklore\footnote{It is important here that all languages considered in the paper are regular.} (see e.g. \cite[Prop. 1.3]{1995:RhodesWeil}) that one may construct a finite monoid $M$ and a surjective morphism $\eta: A^* \to M$ satisfying the following properties.
	\begin{itemize}
		\item Every $Q \in \Qb$ is recognized by $\eta$.
		\item Every language recognized by $\eta$ is a Boolean combination of languages having the form $u\inv Q v\inv$ with $Q \in \Qb$ and $u,v \in A^*$.
	\end{itemize}
	For every $s \in M$, let $P_s = \eta\inv(s)$. We let $\Pb = \{P_s \mid s \in M\}$. Clearly, \Pb is a finite partition of $A^*$ and a congruence. The first assertion implies that \Pb refines \Qb and the second one that \Pb is a \Ds-partition since \Ds is a \vari.
\end{proof}

\subsection{Closure properties of enrichment}\label{sec: closure enrichment}

The definition of enrichment makes sense for any two classes \Cs and \Ds. However, one needs a few hypotheses on \Cs and \Ds for it to be robust. The technical proposition below summarizes the closure properties we will need.

\begin{proposition} \label{prop:closureenrich}
	Let \Cs be a \plivari and \Ds a \livari. Then $\Cs \circ \Ds$ is a \plivari containing \Cs and \Ds. 
\end{proposition}

\begin{proof}
	We fix the \plivari \Cs and the \livari of regular languages \Ds for the proof. There are several properties of $\Cs \circ \Ds$ to show.
	
	\paragraph{\bf Containment.} We first prove that $\Cs \circ \Ds$ contains \Cs and \Ds. Consider an alphabet $A$. We start with \Ds. Let $L \in \Ds(A)$. Since $\Ds$ is a Boolean algebra, $\Pb =\{L,A^* \setminus L\}$ is a \Ds-partition of $A^*$. Moreover $(\Pb \times A)^*$ and $\emptyset$ lie in the lattice $\Cs(\Pb \times A)$. The equality
	\[
	L = \left(\tau_\Pb\inv((\Pb \times A)^*) \cap L\right) \cup\left(\tau_\Pb\inv(\emptyset) \cap \left(A^* \setminus L\right)\right)
	\]
	then shows that $L \in (\Cs \circ \Ds)(A)$.	We turn to \Cs. Consider $L \in \Cs(A)$. Let $\Qb = \{A^*\}$ which is clearly a \Ds-partition of $A^*$. Moreover consider the morphism $\alpha: (\Qb \times A)^* \to A^*$ defined by $\alpha(A^*,a) = a$ for all $a\in A$. Since \Cs is closed under inverse length increasing morphism, we have $\alpha\inv(L) \in \Cs(\Qb \times A)$. It is now clear that
	\[
		L = \tau_{\Qb}\inv(\alpha\inv(L)) \cap A^*.
	\]
	This yields as desired that $L \in (\Cs \circ \Ds)(A)$.
	
	\paragraph{\bf Union and intersection.}	We now prove that $\Cs \circ \Ds$ is a lattice. Consider an alphabet $A$. First we verify that $A^*, \emptyset \in \Cs \circ \Ds$. This follows from the observation that $\Pb = \{A^*\}$ always is a \Ds-partition and $L_{A^*} = A^*$ always sits in the lattice \Cs: then $A^* = \bigcup_{P\in\Pb}\tau_P\inv(L_P) \cap P$ is in $\Cs \circ \Ds$. The same holds for $\emptyset$ if we let $L_{A^*} = \emptyset$. Now let $K,L \in (\Cs \circ \Ds)(A)$. We show that $K \cup L$ and $L \cap K$ belong to $(\Cs \circ \Ds)(A)$ as well. By definition there exist \Ds-partitions $\Pb_K$ and $\Pb_L$ of $A^*$ such that $K$ and $L$ are respectively $\Pb_K$-liftable and $\Pb_L$-liftable from \Cs. Since \Ds is a Boolean algebra, it is immediate that we may construct a third \Ds-partition \Pb of $A^*$ which refines both $\Pb_K$ and $\Pb_L$. By Lemma~\ref{lem:refine}, this implies that $K$ and $L$ are both \Pb-liftable from \Cs. Therefore there exist languages $K_P,L_P \in \Cs(\Pb \times A)$, for every $P \in \Pb$, such that
	\[
	K = \bigcup_{P \in \Pb} \left(\tau_{\Pb}\inv(K_P) \cap P\right) \quad \text{ and } \quad L = \bigcup_{P \in \Pb} \left(\tau_{\Pb}\inv(L_P) \cap P\right).
	\]
	Since inverse images commute with Boolean operations, it follows that
	\[
	K \cup L = \bigcup_{P \in \Pb} \left(\tau_{\Pb}\inv(K_P \cup L_P) \cap P\right)\quad\textrm{ and }\quad K \cap L = \bigcup_{P \in \Pb} \left(\tau_{\Pb}\inv(K_P \cap L_P) \cap P\right).
	\]
	Since \Cs is a lattice, it is immediate that $K_P \cup L_P$ and $K_P \cap L_P$ belong to $\Cs(\Pb \times A)$ for every $P \in \Pb$. Consequently $K \cup L$ and $K \cap L$ are \Pb-liftable from \Cs and belong to $(\Cs \circ \Ds)(A)$.
	
	We turn to quotients and inverse morphism. Let $A,B$ be alphabets, $\alpha: B^* \to A^*$ a length increasing morphism and $u \in A^*$. Consider $L \in (\Cs\circ\Ds)(A)$. We show that $u\inv L$ and $L u\inv$ belong to $(\Cs \circ \Ds)(A)$ and $\alpha\inv(L)$ belongs to $(\Cs \circ \Ds)(B)$. Let us start with an observation.
	
	By hypothesis on \Ds and Lemma~\ref{lem:partmult}, there exists a finite index \Ds-congruence \Pb of $A^*$ such that $L$ is \Pb-liftable from \Cs. It follows that there exist languages $L_P \in \Cs(\Pb \times A)$ for every $P \in \Pb$ such that
	\[
	L = \bigcup_{P \in \Pb} \left(\tau_{\Pb}\inv(L_P) \cap P\right).
	\]
Since \Pb is a congruence, \Pb is a monoid for an operation denoted by $\cmult$, such that the map $w \mapsto \ppart{w}$ is a morphism. If $P \in \Pb$, we denote by $\lambda_P$ the (length increasing) morphism from $(\Pb \times A)^*$ to itself given by $\lambda_P((Q,a)) = (P \cmult Q,a)$ for every $(Q,a)\in \Pb \times A$. One may verify from the definitions that for every $u,v \in A^*$, we have
	\[
	\tau_{\Pb}(uv) = \tau_{\Pb}(u) \lambda_{\ppart{u}}(\tau_{\Pb}(v)).
	\]
	
	\paragraph{\bf Right quotients.} We first show that $L u\inv \in (\Cs \circ \Ds)(A)$. Let $P \in \Pb$ and $v\in P$. Then $v\in Lu\inv$ if and only if $vu\in L$, if and only if $\tau_{\Pb}(vu) \in L_Q$, where $Q = \ppart{vu} = \ppart{v} \cmult \ppart{u} = P \cmult \ppart{u}$. Since $\tau_{\Pb}(vu) = \tau_{\Pb}(v)\ \lambda_{\ppart{v}}(\tau_{\Pb}(u)) = \tau_{\Pb}(v)\ \lambda_{P}(\tau_{\Pb}(u))$, we find that
	$$Lu\inv \cap P = \{v \in A^* \mid \tau_{\Pb}(v) \in L_Q (\lambda_{P}(\tau_{\Pb}(u)))\inv\} \cap P.$$
	Then if we let $M_P = L_{P \cmult \ppart{u}} (\lambda_{P}(\tau_{\Pb}(u)))\inv$ for each $P\in \Pb$, we get
	\[
	Lu\inv = \bigcup_{P \in \Pb} \left(\tau_{\Pb}\inv(M_P) \cap P\right).
	\]
	Since $\Cs$ is quotient-closed, this establishes that $Lu\inv \in (\Cs\circ\Ds)(A)$.
	
	\paragraph{\bf Left quotients.} We turn to the proof that $u\inv L \in (\Cs \circ \Ds)(A)$. Let $P \in \Pb$ and $v\in P$. Then $v \in u\inv L$ if and only if $uv \in L$, if and only if $\tau_{\Pb}(uv) \in L_{\ppart{u} \cmult P}$, if and only if $\tau_{\Pb}(u) \lambda_{\ppart{u}}(\tau_{\Pb}(v)) \in L_{\ppart{u} \cmult P}$. Therefore, if we let $M_P = \lambda_{\ppart{u}}\inv(\tau_{\Pb}(u))\inv (L_{\ppart{u} \cmult P})$, we find that
	\[
	u\inv L = \bigcup_{P \in \Pb} \left(\tau_{\Pb}\inv(M_P) \cap P\right).
	\]
	Since $\Cs$ is quotient-closed and closed under length increasing morphisms, this establishes that $u\inv L \in (\Cs\circ\Ds)(A)$, and concludes the proof for quotients.
	
	\paragraph{\bf Inverse length increasing morphisms.} It remains to show that $\alpha\inv(L) \in (\Cs \circ \Ds)(B)$. Consider the finite partition \Qb of $B^*$ given by
	\[
	\Qb = \{\alpha\inv(P) \mid P \in \Pb \text{ and } \alpha\inv(P) \neq \emptyset\}.
	\]
	Then \Qb is a \Ds-partition since \Pb is one and \Ds is closed under inverse length increasing morphisms. 
	
	We define a morphism $\beta\colon (\Qb \times B)^* \to (\Pb \times A)^*$ as follows. Let $(Q,b)$ be a letter in $\Qb \times B$. By definition of \Qb, $Q = \alpha\inv(P)$ for a uniquely determined $P\in \Pb$ and in particular, $P = \ppart{\alpha(u)}$ for any $u\in Q$. Then we let
	\[
	\beta((Q,b)) = \lambda_P(\tau_{\Pb}(\alpha(b))).
	\]
	Clearly, this defines a morphism $\beta\colon (\Qb \times B)^* \to (\Pb \times A)^*$. Moreover, it is length increasing since both $\lambda_{P}$ and $\alpha$ are.
	
	\begin{fct} \label{fact:definemorph}
		The morphism $\beta$ satisfies $\beta \circ \tau_{\Qb} = \tau_{\Pb} \circ \alpha$. 
	\end{fct}
	
	\begin{proof}
		We verify that $\beta \circ \tau_{\Qb}(w) = \tau_\Pb \circ \alpha(w)$ for every $w \in B^*$, by induction on the length of $w$. The result is trivial if $w = \varepsilon$ and we now assume that $|w| \geq 1$. Then $w = ub$ for some letter $b \in B$ and word $u \in B^*$. By Fact~\ref{fct:tagdecomp}, $\tau_{\Qb}(w) = \tau_{\Qb}(u) \cdot (\qpart{u},b)$ and it follows, by induction, that
		\[
		\beta(\tau_{\Qb}(w)) = \beta(\tau_{\Qb}(u)) \ \beta((\qpart{u},b)) = \tau_{\Pb}(\alpha(u)) \ \beta((\qpart{u},b)).
		\]
		By definition, $\beta((\qpart{u},b)) = \lambda_P(\tau_{\Pb}(\alpha(b)))$, where $P = \ppart{\alpha(u)}$. Therefore
		$$\beta(\tau_{\Qb}(w)) = \tau_{\Pb}(\alpha(u)) \ \lambda_{\ppart{\alpha(u)}}(\tau_{\Pb}(\alpha(b))) = \tau_{\Pb}(\alpha(u) \alpha(b)) =  \tau_{\Pb}(\alpha(ub)) =  \tau_{\Pb}(\alpha(w)),$$
		as desired.
	\end{proof}
	
	Recall that we have languages $L_P \in \Cs(\Pb \times A)$ for every $P \in \Pb$ such that
	\[
	L = \bigcup_{P \in \Pb} \left(\tau_{\Pb}\inv(L_P) \cap P\right).
	\]
	Fact~\ref{fact:definemorph} shows that
	$$\alpha\inv(L) = \bigcup_{P \in \Pb} \left(\alpha\inv(\tau_{\Pb}\inv(L_P)) \cap \alpha\inv(P)\right) = \bigcup_{P \in \Pb} \left(\tau_{\Qb}\inv(\beta\inv(L_P)) \cap \alpha\inv(P)\right).$$
	For each $Q \in \Qb$, we note that $Q= \alpha\inv(P)$ for a uniquely defined $P \in \Pb$, and we let $H_Q = \beta\inv(L_P)$. Since \Cs is closed under inverse (length increasing) morphisms, every $H_Q$ lies in $\Cs(\Qb \times B)$ and we have
	\[
	\alpha\inv(L) = \bigcup_{Q \in \Qb} (\tau_{\Qb}\inv(H_Q) \cap Q),
	\]
	showing directly that $\alpha\inv(L) \in (\Cs \circ \Ds)(B)$. This concludes the proof of Proposition~\ref{prop:closureenrich}.
\end{proof}

\subsection{Enrichment by modulo languages}\label{sec: enrichment and logic}

As mentioned in Remark~\ref{rk: enrich with SU, MOD}, we are mainly interested in the special instances of \Ds-enrichment for two particular classes \Ds. This is because these operations are the language-theoretic counterparts of the logical enrichment operations with local and modular predicates that we introduced in Section~\ref{sec: fragments and logical enrichment}.

We begin with the class \md (\emph{modulo languages}). We first define it and then show that \md-enrichment corresponds exactly to enriching signatures with modular predicates. For every alphabet $A$, $\md(A)$ consists of the finite Boolean combinations of languages of the form $\{w \in A^* \mid |w| = m \mod d\}$, with $m,d \in \nat$ and $m < d$. The following is immediately verified.

\begin{proposition} \label{prop:modvari}
	 The class \emph{\md} is a \vari.
\end{proposition}

\begin{remark}
	\md is not closed inverse morphisms, nor even under inverse length increasing morphisms. Indeed, let $A = \{a,b\}$ and consider the language $L = \{w \in A^* \mid |w| = 0 \mod 2\} \in \md$. Let $\alpha\colon  A^* \to A^*$ be the morphism defined by $\alpha(a) = aa$ and $\alpha(b) = b$. One may verify that $\alpha\inv(L) = \{w \in A^* \mid 2|w|_a + |w|_b = 0 \mod 2\}$ (here $|w|_a$ denotes the number of copies of the letter ``$a$'' occurring in $w$), a language outside \md. It is known that \md is closed under a weaker variant of inverse morphisms: inverse length multiplying morphisms (we shall not need this property).  A morphism $\alpha\colon  A^* \to B^*$ is \emph{length multiplying} when there exists $k \geq 1$ depending only on $\alpha$ such that for every $w \in A^*$, we have $|\alpha(w)| = k |w|$.	
\end{remark}

We complete this definition with the following lemma. It states an elementary, yet useful property of \md.

\begin{lemma} \label{lem:mdnormalform}
	Let $A$ be an alphabet and $L \in \emph{\md}(A)$. There exists a natural number $k \geq 1$ such that for every $w,w' \in A^*$, if $|w|$ and $|w'|$ are congruent modulo $k$, then $w \in L$ if and only if $w' \in L$.
\end{lemma}

\begin{proof}
	By definition, $L$ is a Boolean combination of languages of the form $\{w \in A^* \mid |w| = m \mod d\}$ for some $m,d \in \nat$. It suffices to let $k$ be the least common multiple of these $d$.  
\end{proof}

We now connect \md-enrichment with the logical operation of enriching signatures with modular predicates. For every fragment \Fs and every structural signature \frS, the equality $\Fs(\frS) \circ \md = \Fs(\frS,\MOD)$ holds. This result is essentially folklore. However, the proofs available in the literature only apply to specific fragments and signatures. For example, a proof that $\bswm{1} = \bsw{1} \circ \md$ is available in~\cite{2006:ChaubardPinStraubing-LICS}. Here, we properly establish the generic correspondence.

\begin{theorem} \label{thm:wform:logco}
	Let $\Fs$ be a fragment of first-order logic and let \frS be a structural signature. Then $\Fs(\frS) \circ \emph{\md} = \Fs(\frS,\MOD)$.
\end{theorem}

\begin{proof}
	The two directions in the proof are handled separately.

\paragraph{From $\Fs(\frS) \circ \md$ to $\Fs(\frS,\MOD)$}
Let $L$ be a language in $(\Fs(\frS) \circ \md)(A)$. By definition, there exists a \md-partition \Pb of $A^*$ and a language $L_P \subseteq (\Pb \times A)^*$ in $\Fs(\frS)$ for every $P \in \Pb$ such that
\[
L = \bigcup_{P \in \Pb} \left(\tau_\Pb\inv(L_P) \cap P \right).
\]
We want to show that $L \in \Fs(\frS,\MOD)(A)$. Since $\Fs$ is a fragment of first-order logic, it is closed under conjunctions and disjunctions, and it suffices to prove that for each $P \in \Pb$, both $P$ and $\tau_\Pb\inv(L_P)$ are defined by $\Fs[A,\frS,\MOD]$ sentences. The argument is based on the following lemma which is immediate from the definition of $\md(A)$ and \MOD.

\begin{lemma} \label{lem:prefix}
	For any language $Q \in \emph{\md}(A)$, the following properties holds:
	\begin{enumerate}
		\item $Q$ is defined by a quantifier-free sentence $\psi_Q$ of $\Fs[\MOD]$.
		\item There exists a quantifier-free formula $\xi_Q(x)$ in $\Fs[\MOD]$ with one free variable such that for any $w = a_0 \cdots a_n \in A^*$ and any position $x$ in $w$, $w \models \xi_Q(x)$ if and only if $a_0 \cdots a_{x-1} \in Q$. 
	\end{enumerate}
\end{lemma}

Consider $P \in \Pb$. Since \Pb is a \md-partition, Lemma~\ref{lem:prefix}~(1) shows that $P$ is defined by the sentence $\psi_P$ in $\Fs[\MOD]$, which is contained in $\Fs[A,\frS,\MOD]$.

Now consider the language $\tau_\Pb\inv(L_P) \subseteq A^*$. By definition, $L_P \subseteq (\Pb \times A)^*$ is defined by a sentence $\phi_P$ in $\Fs[\Pb \times A,\frS]$. We build an $F[A,\frS,\MOD]$-sentence $\hat\phi_P$ defining $\tau_\Pb\inv(L_P)$ by applying quantifier-free substitutions to $\phi_P$.

Note that the sentence $\varphi_P$ contains two kinds of atomic formulas: those involving label predicates, of the form $(Q,a)(x)$ for some $(Q,a) \in \Pb \times A$, and those involving structural predicates in \frS. The sentence $\hat\phi_P$ is obtained from $\phi_P$ by replacing each atomic sub-formula $(Q,a)(x)$ in $\phi_P$ ($(Q,a) \in \Pb \times A$) by the quantifier-free $\Fs[A,\MOD]$-formula:
\[
\xi_Q(x) \wedge a(x) \quad \text{where $\xi_Q(x)$ is as given by Lemma~\ref{lem:prefix} (2).}
\]
Then $\hat\phi_P$ belongs to $\Fs[A,\frS,\MOD]$ since every fragment is closed under quantifier-free substitutions. Moreover, one can verify directly from the definition of $\tau_\Pb\colon  A^* \to (\Pb  \times A)^*$ that $\hat\phi_P$ defines $\tau_\Pb\inv(L_P)$. This concludes the proof that $L$ is defined by a sentence in $\Fs(\frS,\MOD)$.

\paragraph{From $\Fs(\frS,\MOD)$ to $\Fs(\frS) \circ \md$}

We now want to show that a language $L$ in $\Fs(\frS,\MOD)$ lies in $\Fs(\frS) \circ \md$. First, we define an appropriate \md-partition \Pb of $A^*$.

Let $\phi$ be a sentence of $\Fs[A,\frS,\MOD]$ defining $L$ and let $p$ be the least common multiple of the integers $d \geq 1$ such that $\varphi$ contains a modular predicate of the form $\MOD^d_j(x)$ or $\D^d_j$ for some $j < d$ ($p$ is well-defined since $\varphi$ contains finitely many predicates; we let $p = 1$ if $\phi$ contains no modular predicate). For every $i < p$, we define
\[
P_i = \{w \in A^* \mid |w| = i \mod p\} \in \md
\]
and we let $\Pb = \{P_i \mid i < p\}$. Clearly \Pb is a \md-partition of $A^*$ and the following fact holds.

\begin{fct} \label{fct:replacemod}
	For any predicate $\MOD^d_j(x)$ occurring in $\varphi$, there exists a quantifier-free formula $\zeta^d_{j}(x)$ of $\Fs[\Pb \times A,\frS]$ such that for any $w \in A^*$ and any position $x \geq 0$ in $w$, we have
	\[
	w \models \MOD^d_j(x) \quad \text{if and only if} \quad \tau_\Pb(w) \models \zeta^d_{j}(x).
	\]
\end{fct}

\begin{proof}
	We let
	\[
	\zeta^d_{j}(x) = \bigvee_{a \in A} \left(\bigvee_{\{i < p \mid j = i\!\mod d\}} (P_i,a)(x)\right).
	\]
	The result follows since $p$ is a multiple of $d$.
\end{proof}

Next we exhibit languages $L_P \subseteq (\Pb \times A)^*$ for every $P \in \Pb$, all in $\Fs(\frS)$ and such that
\[
L = \bigcup_{P \in \Pb} \left(P \cap \tau_\Pb\inv(L_P)\right), 
\]
and for that purpose, we rely on the following lemma.

\begin{lemma} \label{lem:consformula}
	For any $P \in \Pb$, there exists a sentence $\psi_P$ of $\Fs[\Pb \times A,\frS]$ such that for any $w \in P$, the following equivalence holds:
	\[
	w \models \varphi \quad \text{if and only if} \quad \tau_\Pb(w) \models \psi_P.
	\]
\end{lemma}

Let us assume for a moment that Lemma~\ref{lem:consformula} holds. Then, for each $P \in \Pb$, we let $L_P$ be the language in $(\Pb \times A)^*$ defined by the sentence $\psi_P$ provided by Lemma~\ref{lem:consformula}. By definition, $L_P\in \Fs(\frS)$, and the equality
\[
L = \bigcup_{P \in \Pb} \left(P \cap \tau_\Pb\inv(L_P)\right) 
\]
follows immediately from the lemma. We finish with the proof of Lemma~\ref{lem:consformula}.

\begin{proof}[Proof of Lemma~\ref{lem:consformula}]
	By definition $P = P_i$ for some $i < p$ (i.e. $P$ is the language of all words whose length is congruent to $i$ modulo $p$). We build $\psi_P$ from $\phi$ in $\Fs[A,\frS,\MOD]$ by replacing each atomic sub-formula by a quantifier-free $\Fs[\Pb \times A,\frS]$ formula as follows:
	\begin{itemize}
		\item the atomic formula $a(x)$ ($a \in A$) is replaced by
		\[
		\bigvee_{Q \in \Pb} (Q,a)(x);
		\]
		\item an atomic formula involving the structural predicates in \frS remains unchanged;
		\item the atomic formula $\MOD^d_j(x)$ (with $j < d$) is replaced by the formula $\zeta^d_{j}(x)$ given by Fact~\ref{fct:replacemod};
		\item the atomic formula $\D^d_j$ (with $j < d$) is replaced by $\top$ if $i$ and $j$ are congruent modulo $d$ and $\bot$ otherwise.
	\end{itemize}
	Since $P = P_i = \{w \in A^* \mid |w| = i \mod p\}$ and $p$ is a multiple of $d$, if $j < p$, then either every word in $P$ satisfies $\D^d_j$ (exactly when $i$ and $j$ are congruent modulo $d$) or no word in $P$ does. It is now straightforward to verify that $\psi_P$ satisfies the desired property: for any $w \in P$, $w \models \varphi$ if and only if $\tau_\Pb(w) \models \psi_P$.	
\end{proof}
This concludes the proof of Theorem~\ref{thm:wform:logco}.
\end{proof}

\subsection{Enrichment by suffix languages}

We now consider enrichment of logical signatures by local predicates. In most cases, it corresponds to \su-enrichment where \su denotes the class of suffix languages, defined below.

\begin{remark}
	There is a significant difference with what happened for \md-enrichment and modular predicates. The correspondence stated in Theorem~\ref{thm:wform:logco} is \emph{generic}. This is not the case here. It is true that in almost all relevant cases, the equality $\Fs(\frS) \circ \su = \Fs(\frS,+1)$ holds. However, this is not a generic theorem and there are counter-examples. For example, consider $\fo(\emptyset)$ (i.e. the structural signature is empty and only the label predicates are available). It turns out that the class $\fo(+1)$ is strictly larger than $\fo(\emptyset) \circ \su$.
	
	In practice, establishing the equality $\Fs(\frS) \circ \su = \Fs(\frS,+1)$ requires a proof that is specific to $\Fs(\frS)$, and often technical and tedious. Fortunately, this has already been achieved for the classes that we consider in the paper.
\end{remark}

Let us first define \su. For every alphabet $A$, $\su(A)$ consists of the finite Boolean combinations of languages of the form $A^*w$, for some $w \in A^*$. These languages are sometimes called \emph{definite}, see \emph{e.g.}, \cite{1976:Eilenberg,1985:Straubing}. Again, the following is folklore and elementary.

\begin{proposition} \label{prop:suvar}
	The class \emph{\su} is an \livari.
\end{proposition}

We consider a natural stratification $(\su_k)_k$ within \su: for every $k \in \nat$ and every alphabet $A$, we let $\su_k(A)$ be the set of Boolean combinations of languages of the form $A^*w$, where $w \in A^*$ and $|w| \leq k$. Note that every $\su_k(A)$ is finite. Each stratum $\su_k$ is an \livari. Moreover, for every alphabet $A$, we have
\[
\su_k(A) \subseteq \su_{k+1}(A) \text{ for every $k \in \nat$} \qquad \text{and} \qquad \bigcup_{k \in \nat} \su_k(A) = \su(A).
\]
Given an alphabet $A$ and an integer $k\in\nat$, the equivalence relation $\keqsu$ on $A^*$ is defined as follows if $w,w' \in A^*$, we let $w \keqsu w'$ if and only if the following condition holds:
\[
\text{For every language $L \in \su_k(A)$,} \quad w\in L \Leftrightarrow w' \in L.
\]
In other words, $w \keqsu w'$ if and only if $w = w'$ or $|w|, |w'| \ge k$ and $w$ and $w'$ have the same length $k$ suffix. Note that \keqsu has finite index since $\su_k$ is finite. Since every stratum $\su_k$ is a Boolean algebra, the following lemma is immediate.

\begin{lemma} \label{lem:canoeq}
	Let $k \in \nat$. For every alphabet $A$, the languages in $\emph{\su}_k(A)$ are exactly the unions of \keqsu-classes.
\end{lemma}

%

As announced, it is known that for important fragments of first-order, \su-enrichment corresponds the logical operation of enrichment by local predicates. This goes back to Straubing \cite{1985:Straubing} for the quantifier alternation hierarchies, and to Thérien and Wilke \cite{1998:TherienWilke} for \fod. Place and Zeitoun reformulated these results in terms of language class enrichment in \cite[Prop. 5.2 and 6.2]{2017:PlaceZeitoun}.

\begin{theorem}\label{thm: su for classical fragments}
Let $\Fs$ be one of the following fragments of first-order logic: \emph{\fod}, \sic{n}, \bsc{n} ($n\ge 1$). Then $\Fs(<) \circ \emph{\su} = \Fs(<,+1)$.
\end{theorem}

Another useful result is that applying \su-enrichment to \fow does not build a larger class. This not surprising: as we explained in Remark~\ref{rem:fosucc}, the classes \fow and \fows coincide. We prove this in the following proposition.

\begin{proposition} \label{prop: su for fo}
	We have $\emph{\fow} \circ \emph{\su} = \emph{\fow}$.
\end{proposition}

\begin{proof}
	Since \fow is a \varie, it is immediate from Proposition~\ref{prop:closureenrich} that \fow is contained in $\fow \circ \su$. We now verify the converse inclusion. Let $A$ be an alphabet and consider $L \in (\fow \circ \su)(A)$. We show that $L \in \fow$. By definition, there exists an \su-partition \Pb of $A^*$ such that $L$ is \Pb-liftable from \fow. Thus, we have languages $L_P \in \fow(\Pb \times A)$ for every $P \in \Pb$ such that
	\[
	L = \bigcup_{P \in \Pb} \tau_{\Pb}\inv(L_P) \cap P.
	\]
	Since \fow is closed under union and intersection, it suffices to show that for each $P$ in \Pb, both $P$ and $\tau_{\Pb}\inv(L_P)$ are in \fow. We use the following fact.
	
	\begin{fct} \label{fct:fosucc}
		For every $K \in \emph{\su}(A)$, one may build the following $\emph{\fo}[A,<]$ formulas:
		\begin{itemize}
			\item a sentence $\varphi_K$ which defines $K$;
			\item a formula $\psi_K(x)$ with one free variable $x$ such that, for every word $u \in A^*$ and every position $i$ in $u$, $u \models \psi_K(i)$ if and only if the prefix of $u$ ending at position $i-1$ belongs to $K$.
		\end{itemize}
	\end{fct}
	
	\begin{proof}
		By definition a language $K \in \su(A)$ is a Boolean combination of languages having the form $A^*w$ for some $w \in A^*$. Hence, since one may freely use Boolean connectives in first-order formulas, it suffices to consider the case when $K$ itself is of this form, $K = A^*w$. Let $a_1,\dots,a_n \in A$ such that $w = a_1 \cdots a_n$. Recall that we may also freely use the successor predicate in $\fo[A,<]$ formulas since $x+1 = y$ is equivalent to $x < y \wedge \neg \left(\exists z\ x < z \wedge z <y\right)$. We may use \textsf{max} as well, as $\textsf{max}(x)$ is equivalent to $\neg \left(\exists y\ x < y\right)$. We define $\varphi_K$ as the following sentence:
		\[
		\exists x_1 \cdots \exists x_n\quad \textsf{max}(x_n) \wedge \left(\bigwedge_{1 \leq i \leq n} a_i(x_i)\right) \wedge \left(\bigwedge_{1 \leq i \leq n-1} x_i +1 = x_{i+1}\right).
		\]
		Furthermore, we define $\psi_K(x)$ as the following formula:
		\[
		\exists x_1 \cdots \exists x_n\quad \left(x_n+1 = x\right) \wedge \left(\bigwedge_{1 \leq i \leq n} a_i(x_i)\right) \wedge \left(\bigwedge_{1 \leq i \leq n-1} x_i +1 = x_{i+1}\right).
		\]
		One may verify that these formulas have the expected semantics. This concludes the proof of Fact~\ref{fct:fosucc}.		
	\end{proof}
	
	We may now finish the proof. Consider $P \in \Pb$. Since \Pb is an \su-partition, we have $P \in \su$ and it is immediate that $P \in \fow$: it is defined by sentence $\varphi_P$ given by Fact~\ref{fct:fosucc}. It remains to show that $\tau_{\Pb}\inv(L_P) \in \fow$. By hypothesis, $L_P \in \fow(\Pb \times A)$. Thus, it is defined by a sentence $\zeta$ of $\fo[\Pb \times A,<]$. It is immediate, from the definition of $\tau_{\Pb}$, that $\tau_{\Pb}\inv(L_P)$ is defined by the formula $\zeta'$ of $\fo[A,<]$ obtained by replacing every occurrence of an atomic formula $(Q,a)(x)$ in $\zeta$ by the formula
	\[
	\psi_Q(x) \wedge a(x) \quad \text{where $\psi_Q(x)$ is the formula given by Fact~\ref{fct:fosucc}}.
	\]
	This implies that $\tau_{\Pb}\inv(L_P) \in \fow$, finishing the proof.
\end{proof}

\section{Main theorem}\label{sec: main theorem}

We are now ready to present our main theorem. As announced, it states a generic reduction for the covering problem which applies to classes of the form $(\Cs \circ \su) \circ \md$ where \Cs is a \pvarie. 

\begin{theorem}\label{thm:main}
	Let \Cs be a \pvarie. The covering (resp. separation) problem for $(\Cs \circ \emph{\su}) \circ \emph{\md}$ can be effectively reduced to the same problem for $\Cs \circ \emph{\su}$.
\end{theorem}

Theorem~\ref{thm:main} can be combined with another reduction theorem for classes of the form $\Cs \circ \su$ (\cite{2015:PlaceZeitoun} and \cite[Thm 4.12]{2017:PlaceZeitoun}, see also \cite{2001:Steinberg} for an algebraic statement relating to the Boolean algebra case): for any \pvarie \Cs, the covering problem for $\Cs \circ \su$ can be effectively reduced to the same problem for \Cs. When combining these theorems, we obtain the following corollary.

\begin{corollary} \label{cor:main}
	Let \Cs be a \pvarie. The covering (resp. separation) problem for $(\Cs \circ \emph{\su}) \circ \emph{\md}$ can be effectively reduced to the same problem for \Cs.
\end{corollary}

The remaining sections of the paper are devoted to the proof of Theorem~\ref{thm:main}. It is organized as follows. First we show (Theorem~\ref{thm:blockenr}) that, for classes of the form $\Cs \circ \su$, \md-enrichment can be reformulated as yet another operation called block abstraction, written $\Ds \mapsto \benr{\Ds}$, which is much simpler to handle. Next we show (in Section~\ref{sec:transfer}) that for any class \Ds closed under inverse length increasing morphisms, the $\benr{\Ds}$-covering problem reduces to the \Ds-covering problem. The combination of these two properties establishes Theorem~\ref{thm:main}: for any \pvarie \Cs, Proposition~\ref{prop:closureenrich} ensures that $\Cs \circ \su$ is closed under inverse length increasing morphisms, and hence the covering problem for $(\Cs \circ \su) \circ \md$ reduces to the same problem for $\Cs \circ \su$.

However, before we start this proof, we discuss applications of Theorem~\ref{thm:main} and Corollary~\ref{cor:main}: combining them with previously known results yields decidability results for new fragments.

\subsection{Full first-order logic} Let us start with \fo itself. We showed in Proposition~\ref{prop: su for fo} that \fow is stabilized by \su-enrichment: $\fow = \fow \circ \su$. Moreover, Theorem~\ref{thm:wform:logco} shows that $\fowm = \fow \circ \md$ and, since \fow is a \varie, we may instantiate Theorem~\ref{thm:main} to obtain an effective reduction from \fowm-covering to \fow-covering.

Finally, \fow-covering is known to be decidable (see, \textit{e.g.}, \cite{pzfoj}). Altogether, we get the following corollary.

\begin{corollary} \label{cor:fo}
	The covering problem for \emph{\fowm} is decidable.
\end{corollary}

\subsection{Two-variable first-order logic}
We now consider \fod. As stated in Theorem~\ref{thm: su for classical fragments}, we have $\fodws = \fodw \circ \su$. When combined with Theorem~\ref{thm:wform:logco}, this yields $\fodwsm = (\fodw \circ \su) \circ \md$. Since \fodw is a variety, Corollary~\ref{cor:main} gives an effective reduction from \fodwsm-covering to \fodw-covering.

It was shown in~\cite{2016:PlaceZeitoun,2018:PlaceZeitoun} that \fodw-covering is decidable (see also~\cite{pvzmfcs13} for a proof restricted to \fodw-separation). We obtain the following corollary.

\begin{corollary} \label{cor:fo2}
	The covering problem for \emph{\fodwsm} is decidable. 
\end{corollary}


\subsection{Quantifier alternation hierarchy}

Finally, we consider the quantifier alternation hierarchy of first-order logic. We may combine Theorems~\ref{thm: su for classical fragments} and~\ref{thm:wform:logco} to obtain that for every $n \geq 1$,
\[
\siwsm{n} = (\siw{n} \circ \su) \circ \md \text{ and } \bswsm{n} = (\bsw{n} \circ \su) \circ \md.
\]
Since \siw{n} and \bsw{n} are respectively a \pvarie and a \varie, Corollary~\ref{cor:main} shows that \siwsm{n}-covering (resp. \bswsm{n}-covering) is effectively reducible to \siw{n}-covering (resp. \bsw{n}-covering).

It is shown in~\cite{2014:PlaceZeitoun,pseps3,pseps3j} that \siw{1}-, \siw{2}- and \siw{3}-covering are decidable. Moreover, \bsw{1}- and \bsw{2}-covering are proved to be decidable in~\cite{pzboolpol} (there also exists many proofs for the decidability of \bsw{1}-separation, see~\cite{pvzmfcs13,2013:CzerwinskiMartensMasopust}). Altogether, we get the following corollary.

\begin{corollary} \label{cor:main2}
	The covering problem for the levels \siwsm{1}, \bswsm{1}, \siwsm{2}, \bswsm{2} and \siwsm{3} is decidable.
\end{corollary}

%

\section{Block abstraction}
\label{sec:block}
In this section, we introduce the operation $\Cs \mapsto \benr{\Cs}$ of \emph{block abstraction}, defined on classes of languages. While less standard than enrichment, this operation is a key ingredient of this paper because of its connection with \md-enrichment: the two operations coincide for classes of the form $\Cs \circ \su$ where \Cs is a \pvarie.

\subsection{Definition of block abstraction}\label{sec: definition block abstraction}

Let $d \geq 1$ be an integer and $A$ an alphabet. Recall that $A^d$ denotes the set of words with length $d$. We let $A^{<d} = \bigcup_{0<c<d}A^c$ and $A_d = A^{<d} \cup A^d$. Our intent is to use $A_d$ as an alphabet and form words in $(A_d)^*$. This can be misleading since a letter $w \in A_d$ is also a word in $A^*$. To avoid confusion, we adopt the following convention: when $w\in A_d$ is used as a letter of the alphabet $A_d$, we denote it by $(w)$.

We define a map $\mu_d\colon A^* \to A_d^*$ as follows. If $w \in A^*$ has length $\ell$, we consider the Euclidean division of $\ell$ by $d$: $\ell = d q + r$, with $q,r \in \nat$ and $0 \leq r < d$. Then $w$ admits a decomposition $w = w_1 \cdots w_{q+1}$ where $|w_i| = d$ for $1 \leq i \leq q$ and $|w_{q+1}| = r$, and we let
\[
\mu_d(w) = \left\{\begin{array}{ll}
(w_1)\cdots (w_q) & \text{if $w_{q+1} = \varepsilon$ (i.e. $r = 0$)}, \\
(w_1)\cdots (w_q) (w_{q+1}) & \text{if $w_{q+1} \neq \varepsilon$ (i.e. $r \neq 0$)}.
\end{array}\right.
\]

\begin{remark}
	The map $\mu_d\colon A^* \to A_d^*$ is not a morphism and it is not surjective: all letters in $\mu_d(w)$, except possibly the last one, are in $A^d$. However, $\mu_d$ is clearly injective.
\end{remark}

We also note the following elementary fact.

\begin{fct} \label{fct:blockcut}
	Let $d \geq 1$ and $u,v \in A^*$. If the length of $u$ is a multiple of $d$, then $\mu_d(uv) = \mu_d(u) \cdot \mu_d(v)$.
\end{fct}

Let \Cs be a class of languages. For any natural number $d \geq 1$, we let \benrd{\Cs} be the class such that, for any alphabet $A$, $\benrd{\Cs}(A)$ consists of the languages $L$ of the form
\[
L = \mu_d\inv(K) \quad \text{for some language $K \in \Cs(A_d)$.}
\]
The \emph{block abstraction} of \Cs is defined as the union of the classes \benrd{\Cs} for $d \geq 1$. More precisely, for any alphabet $A$, we let
\[
\benr{\Cs}(A) = \bigcup_{d \geq 1} \benrd{\Cs}(A).
\]
We record a few simple properties of the block abstraction operation.

\begin{lemma}\label{lem:blockinc}
	Let \Cs be a class of languages closed under inverse length increasing morphisms. Let $d,n \geq 1$ such that $n$ is a multiple of $d$. Then $\benrd{\Cs} \subseteq \benrp{\Cs}{n}$.
\end{lemma}

\begin{proof}
If $L \in \benrd{\Cs}$, then $L =\mu_d\inv(K)$ for some language $K \in \Cs(A_{d})$. Now let $\eta\colon A_n^* \to A_{d}^*$ be the length increasing morphism which maps each letter $(w) \in A_n$ to $\mu_{d}(w)$. Note in particular that $\mu_d =  \eta\circ\mu_n$ (this is where we use the fact that $n$ is a multiple of $d$).  By our hypothesis on \Cs, the language $H = \eta\inv(K)$ belongs to $\Cs(A_n)$ and we have 
	\[
	L =  \mu_d\inv(K) = \mu_n\inv(\eta\inv(K)) = \mu_n\inv(H),
	\]
	which concludes the proof.
\end{proof}

\begin{corollary}\label{cor:commonindex}
	Let \Cs be a class closed under inverse length increasing morphisms and let $L_1, \ldots, L_m$ ($m \ge 1$) be languages in $\benr{\Cs}$. There exists $d \geq 1$ such that $L_1, \ldots, L_m \in \benrd{\Cs}$.
\end{corollary}

\begin{proof}
	By definition, each $L_i$ ($1\le i \leq m$) lies in $\benrp{\Cs}{d_i}$ for some $d_i \geq 1$, and Lemma~\ref{lem:blockinc} then shows that $L_1, \dots, L_m \in \benrd{\Cs}$ where $d$ is the least common multiplier of the $d_i$.
\end{proof}

This in turn shows that block abstraction preserves closure under Boolean operations.

\begin{corollary} \label{cor:nopclos}
	If \Cs is a lattice closed under inverse length increasing morphisms, then its block abstraction \benr{\Cs} is a lattice.
\end{corollary}

\begin{proof}
	Let $A$ be an alphabet. It is immediate that $\emptyset$ and $A^*$ are elements of $\benr{\Cs}(A)$. Now let $L_1,L_2 \in \benr{\Cs}(A)$. By Corollary~\ref{cor:commonindex}, there exists $d \geq 1$ such that $L_1,L_2 \in \benrd{\Cs}(A)$, so that $L_1 = \mu_d\inv(K_1)$ and $L_2 = \mu_d\inv(K_2)$ for some $K_1,K_2 \in \Cs(A_d)$. Therefore $L_1 \cup L_2 = \mu_d\inv(K_1 \cup K_2)$ and $L_1 \cap L_2 = \mu_d\inv(K_1 \cap K_2)$, which concludes the proof since $L_1 \cup L_2$ and $L_1 \cap L_2$ both belong to \benr{\Cs}.
\end{proof}

We now turn to the main theorem of the section.

\begin{theorem} \label{thm:blockenr}
	If \Cs is a \pvarie, then $(\Cs \circ \emph{\su}) \circ \emph{\md} = \benr{\Cs \circ \emph{\su}}$.
\end{theorem}

The rest of this section is devoted to the proof of Theorem~\ref{thm:blockenr}. We fix a \pvarie \Cs and an alphabet $A$, and we prove separately that $((\Cs \circ \su) \circ \md)(A) \subseteq \benr{\Cs \circ \su}(A)$ and $\benr{\Cs \circ \su}(A) \subseteq ((\Cs \circ \su) \circ \md)(A)$.

\subsection{From modulo enrichment to block abstraction}

We start with the containment $(\Cs \circ \su) \circ \md \subseteq \benr{\Cs \circ \su}$. We actually prove the following stronger result.

\begin{proposition}\label{prop:mod2enr}
If \Ds is a lattice closed under inverse length increasing morphisms and containing \emph{\su}, then $\Ds \circ \emph{\md} \subseteq \benr{\Ds}$.
\end{proposition}

With reference to Theorem~\ref{thm:blockenr}, we note that if \Cs is a \pvarie, then $\Ds = \Cs \circ \su$ satisfies the hypothesis of Proposition~\ref{prop:mod2enr}, see Proposition~\ref{prop:closureenrich}.

Let us now prove Proposition~\ref{prop:mod2enr}: let \Ds be a lattice closed under inverse length increasing morphisms and containing \su, and let $L \in (\Ds \circ \md)(A)$. Then there exists a \md-partition \Pb of $A^*$ and languages $L_P \in \Ds(\Pb\times A)$ for every $P \in \Pb$ such that
\[
L = \bigcup_{P \in \Pb} \left(\tau_\Pb\inv(L_P) \cap P\right).
\]
Recall that \benr{\Ds} is a lattice by Corollary~\ref{cor:nopclos}. To show that $L\in \benr{\Ds}(A)$, we only need to show that, for every $P \in \Pb$, both $P$ and $\tau_\Pb\inv(L_P)$ belong to $\benr{\Ds}(A)$.

The proof that $P \in \benr{\Ds}(A)$ is handled by the following lemma.

\begin{lemma}\label{lm:mod in block enrichment}
If \Ds is a lattice closed under inverse length increasing morphisms and containing \emph{\su}, then $\emph{\md} \subseteq \benr{\Ds}$.
\end{lemma}

\begin{proof}
Let $P \in \md(A)$. In view of Lemma~\ref{lem:mdnormalform}, there exists a natural number $d \geq 1$ such that $P$ is a finite union of languages of the form $P_{m,d} = \{w \in A^* \mid |w| = m \mod d\}$, with $0 \le m < d$. Since \benr{\Ds} is a lattice, we only need to show that every $P_{m,d}$ is in $\benr{\Ds}(A)$.

If $m = 0$, let $K$ be the set of letters of $A_d$ of the form $(w)$ with $w \in A^d$. Then $P_{m,d} = \mu_d\inv(A_d^*K \cup \{\varepsilon\})$. Since $A_d^*K \cup \{\varepsilon\}$ belongs to $\su(A_d)$ and therefore to $\Ds(A_d)$ by hypothesis on \Ds. It follows that $P_{m,d} \in \benrd{\Ds}$ as desired.

If $m > 0$, we let similarly $K$ be the set of letters of $A_d$ of the form $(w)$ with $w\in A^m$. Then $P_{m,d} = \mu_d\inv(A_d^*K)$ and, since $A_d^*K\in \su(A_d)$, we have $P_{m,d} \in \benrd{\Ds}(A)$.
\end{proof}

To conclude the proof, we show that for every \md-partition $\Pb$ and every language $K \in \Ds(\Pb\times A)$, we have $\tau_\Pb\inv(K) \in \benr{\Ds}(A)$. We first establish the following lemma.

\begin{lemma} \label{lem:choiced}
	There exists a natural number $d \geq 1$ such that for any two words $w,w' \in A^*$, if $|w| = 0 \mod d$, then $\tau_\Pb(ww') = \tau_\Pb(w) \tau_\Pb(w')$.
\end{lemma}

\begin{proof}
By Lemma~\ref{lem:mdnormalform}, for each $P \in \Pb$, there exists $d_P \geq 1$ such that for any $u,v \in A^*$, if $|u| = |v| \mod d_P$, then $u \in P$ if and only if $u \in P$. Using the finiteness of \Pb, we let $d = \lcm\{d_P \mid P \in \Pb\}$. Then if two words $u$ and $v$ satisfy $|u| = |v| \mod d$, we have $\ppart{u} = \ppart{v}$.

Now suppose that $|w| = 0 \mod d$ and that $w' = a_1 \cdots a_n$ with each $a_i \in A$. Using Fact~\ref{fct:tagdecomp} recursively, we obtain,
\[
\tau_\Pb(ww') = \tau_\Pb(w) (\ppart{w},a_1)(\ppart{wa_1},a_2) \cdots  (\ppart{wa_1 \cdots a_{n-1}},a_n).
\]
Moreover, since $|w| = 0 \mod d$, our choice of $d$ implies that $\ppart{wa_1 \cdots a_{i}} = \ppart{a_1 \cdots a_{i}}$ for every $i \leq n-1$. Therefore, we have
\[
\tau_\Pb(ww') = \tau_\Pb(w) (\ppart{\varepsilon},a_1)(\ppart{a_1},a_2) \cdots  (\ppart{a_1 \cdots a_{n-1}},a_n) = \tau_\Pb(w) \tau_\Pb(w'),
\]
which concludes the proof.
\end{proof}

We now show that if $K \in \Ds(\Pb\times A)$ and $d$ is given by Lemma~\ref{lem:choiced}, then $\tau_\Pb\inv(K) \in \benrd{\Ds}(A) \subseteq \benr{\Ds}(A)$. Let $\alpha\colon A_d^* \to (\Pb \times A)^*$ be the morphism given as follows: for each letter $b = (w)\in A_d$ (where $w \in A^+$ is a nonempty word of length at most $d$), let $\alpha(b) = \tau_\Pb(w)$.

This morphism satisfies $\tau_\Pb = \alpha \circ \mu_d$. Indeed, if $u\in A^*$, we consider the decomposition $u = w_1 \cdots w_n$ such that $|w_i| = d$ for every $i\le n-1$ and $|w_n| < d$. Then $\mu_d(w) = (w_1) \cdots (w_n)$ and $\alpha\circ\mu_d(u) = \tau_\Pb(w_1)\cdots\tau_\Pb(w_{n-1})\tau_\Pb(w_n)$. Lemma~\ref{lem:choiced} shows then that $\alpha\circ\mu_d(u) = \tau_\Pb(u)$.

Note now that $|\tau_\Pb(w)| = |w|$ by definition, so $\alpha$ is length increasing. It follows that $\alpha\inv(K) \in\Ds(A_d)$ and hence $\tau_\Pb\inv(K) = \mu_d\inv(\alpha\inv(K)) \in \benrd{\Ds}(A)$, as announced.

\subsection{From block abstraction to modulo enrichment}

The proof that $\benr{\Cs \circ \su} \subseteq (\Cs \circ \su) \circ \md$ is more involved. Consider a language $L \in \benr{\Cs \circ \su}(A)$.  By definition, $L \in \benrd{\Cs \circ \su}(A)$ for some integer $d \geq 1$, which we fix for the remainder of the proof. In particular, we let $\Mb = \{M_i \mid 0\le i < d\}$ be the \md-partition of $A^*$ which counts the length of words modulo $d$. That is, for any integer $0 \le i < d$,
$$M_i = \{w \in A^* \mid |w| = i \mod d\}.$$

We reduce the problem to proving the following proposition.

\begin{proposition} \label{prop:whattodo}
	There exists a language $H \in (\Cs \circ \emph{\su})(\Mb \times A)$ such that $L =  \tau_\Mb\inv(H)$.
\end{proposition}

Indeed, assuming temporarily Proposition~\ref{prop:whattodo}, we see that
\[
L = \bigcup_{M \in \Mb}  \left(\tau_\Mb\inv(H) \cap M\right)
\]
since \Mb is a partition of $A^*$. It then follows from the definition that $L \in ((\Cs \circ \su) \circ \md)(A)$ as desired.

We now focus on proving Proposition~\ref{prop:whattodo}. We start by introducing a convenient definition: given a word $w \in A^*$, the \emph{$d$-cut of $w$} is the unique pair of words $(u,v) \in (A^*)^2$ such that $w = uv$, $|u| = 0 \mod d$ and $0 \leq |v| \leq d-1$ (in particular, $u$ is the longest prefix of $w$ with length a multiple of $d$). The following facts hold by definition of the maps $\mu_d\colon A^* \to A_d$ and $\tau_\Mb\colon A^* \to (\Mb \times A)^*$.

\begin{fct} \label{fct:cuts}
	Let $w \in A^*$ and let $(u,v)$ be its $d$-cut. Then either $v = \varepsilon$ and $\mu_d(w) = \mu_d(u)$, or $v \neq \varepsilon$ and $\mu_d(w) = \mu_d(u) \ (v)$. Moreover, $\tau_{\Mb}(w) = \tau_{\Mb}(u) \ \tau_{\Mb}(v)$.
\end{fct}

By definition of $\Cs \circ \su$, the construction of the language $H \in (\Cs \circ \su)(\Mb \times A)$ announced in Proposition~\ref{prop:whattodo} requires first choosing an \su-partition of $(\Mb\times A)^*$.

Since $L \in \benrd{\Cs \circ \su}(A)$, there exists $K \in (\Cs \circ \su)(A_d)$ such that $L = \mu_d\inv(K)$ and, by definition of \su-enrichment, there exists an \su-partition \Pb of $A_d^*$ and languages $K_P\in \Cs(\Pb \times A_d)$ for every $P \in \Pb$ such that
\[
K = \bigcup_{P \in \Pb} \left(\tau_\Pb\inv(K_P) \cap P\right) \subseteq A_d^*.
\]
\Pb contains finitely many languages in $\su(A_d)$, so there exists $\ell \in \nat$ such that every $P \in \Pb$ belongs to $\su_{\ell}(A_d)$. We then let $k = d (\ell+1)$ and \Ub be the partition of $(\Mb \times A)^*$ into \keqsu-classes. By Lemma~\ref{lem:canoeq}, \Ub is an \su-partition.

To construct $H \in (\Cs \circ \su)(\Mb \times A)$ such that $L =  \tau_\Mb\inv(H)$, we use the following two lemmas, whose proofs are deferred to Sections~\ref{proof of 5.12} and~\ref{proof of 5.13}.

\begin{lemma} \label{lem:thecut}
	Let $w_1,w_2 \in A^*$ and $(u_1,v_1),(u_2,v_2) \in (A^*)^2$ be their $d$-cuts. Assume that $\tau_\Mb(w_1) \keqsu \tau_\Mb(w_2)$. Then $\ppart{\mu_d(w_1)} = \ppart{\mu_d(w_2)}$, $\ppart{\mu_d(u_1)} = \ppart{\mu_d(u_2)}$ and $v_1 = v_2$
\end{lemma}

\begin{lemma} \label{lem:themorph}
	There exists a morphism $\alpha\colon (\Ub \times (\Mb \times A))^* \to (\Pb \times A_d)^*$ such that for any word $w \in A^*$ with $d$-cut $(u,v)\in (A^*)^2$, we have $\tau_\Pb(\mu_d(u))  = \alpha(\tau_\Ub(\tau_\Mb(w)))$.
\end{lemma}

For each $U \in \Ub$, we construct a language $H_U \in \Cs(\Ub \times (\Mb \times A))$ as follows. If $U$ does not intersect $\tau_\Mb(A^*)$, we let $H_U = \emptyset$. Otherwise, it follows from Lemma~\ref{lem:thecut} that there exists $P,Q \in \Pb$ and $x \in A^*$ of length at most $d-1$ such that for any $w \in A^*$ whose $d$-cut is $(u,v) \in (A^*)^2$, if $\tau_\Mb(w) \in U$, then $\ppart{\mu_d(w)} = P$, $\ppart{\mu_d(u)} = Q$ and $v = x$. The definition of $H_U$ in that case uses the morphism $\alpha$ in Lemma~\ref{lem:themorph} and depends on whether $x = \varepsilon$ or $x \neq \varepsilon$. Note that if $x \neq \varepsilon$, then $(x)$ is a letter of $A_d$ and $(Q,(x))$ is a letter of $\Pb \times A_d$. We let
\[
H_U = \left\{\begin{array}{ll}
\alpha\inv(K_P \cdot (Q,(x))\inv) & \text{when $x \neq \varepsilon$} \\
\alpha\inv(K_P) & \text{when $x = \varepsilon$}
\end{array}\right.
\]  
Since \Cs is closed under right quotients and inverse morphisms, and since $K_P \in \Cs(\Pb \times A_d)$, we have $H_U \in \Cs(\Ub \times (\Mb \times A))$. We then define a language $H$ on alphabet $\Mb \times A$ as follows:
\[
H = \bigcup_{U \in \Ub} \left(\tau_\Ub\inv(H_U) \cap U\right).
\]
Since \Ub is an \su-partition of $(\Mb \times A)^*$, the language $H$ is in $(\Cs \circ \su)(\Mb \times A)$. 

We now show that $L = \tau_\Mb\inv(H)$. Since $L = \mu_d\inv(K)$, this amounts to proving that for any $w \in A^*$, the following equivalence holds:
\begin{equation} \label{eq:eqgoal}
\mu_d(w) \in K \quad  \text{if and only if} \quad \tau_\Mb(w) \in H
\end{equation}

So let us fix $w \in A^*$ and let $(u,v) \in (A^*)^2$ be its $d$-cut. Since \Ub is a partition, let $U$ be the unique element of $\Ub$ such that $\tau_\Mb(w) \in U$, and let $P = \ppart{\mu_d(w)}$ and $Q = \ppart{\mu_d(u)}$. By Lemma~\ref{lem:thecut}, $P$, $Q$ and $v$ are entirely determined by $U$. By definition of $K$ and $H$, proving~\eqref{eq:eqgoal} is equivalent to proving
$$\tau_\Pb(\mu_d(w)) \in K_P \quad  \text{if and only if} \quad \tau_\Ub(\tau_\Mb(w)) \in H_U.$$
There are two cases depending on whether $v$ is empty or not. We handle the case when $v \neq \varepsilon$, the other case is similar. By definition of $H_U$, we want to prove
\begin{equation} \label{eq:eqgoal2}
\tau_\Pb(\mu_d(w)) \in K_P \quad  \text{if and only if} \quad \alpha(\tau_\Ub(\tau_\Mb(w))) \in K_P \cdot (Q,v)\inv.
\end{equation}
By Fact~\ref{fct:cuts}, we have $\tau_\Pb(\mu_d(w))= \tau_\Pb(\mu_d(u) \cdot (v))$ and Fact~\ref{fct:tagdecomp} shows that 
\[
\tau_\Pb(\mu_d(w)) = \tau_\Pb(\mu_d(u)) \cdot (\ppart{\mu_d(u)},(v)) = \tau_\Pb(\mu_d(u)) \cdot (Q,(v)).
\]
Therefore $\tau_\Pb(\mu_d(w)) \in K_P$ if and only if $\tau_\Pb(\mu_d(u)) \in K_P \cdot (Q,(v))\inv$. Finally, Lemma~\ref{lem:themorph} allows us to conclude since it shows that $\tau_\Pb(\mu_d(u)) = \alpha(\tau_\Ub(\tau_\Mb(w)))$. This completes the proof of Proposition~\ref{prop:whattodo}.

\subsection{Proof of Lemma~\ref{lem:thecut}}\label{proof of 5.12}

Let $w_1,w_2 \in A^*$ and $(u_1,v_1),(u_2,v_2) \in (A^*)^2$ be their $d$-cuts. Assume that $\tau_\Mb(w_1) \keqsu \tau_\Mb(w_2)$. We want to show that $\ppart{\mu_d(w_1)} = \ppart{\mu_d(w_2)}$, $\ppart{\mu_d(u_1)} = \ppart{\mu_d(u_2)}$ and $v_1 = v_2$. We first establish two facts.

\begin{fct} \label{fct:congtau}
	The lengths $|w_1|$ and $|w_2|$ are congruent modulo $d$.
\end{fct}

\begin{proof}
	By definition of \Mb, the lengths modulo $d$ of $w_1$ and $w_2$ are encoded in the rightmost letters of their $\tau_\Mb$-images, which are identical since $\tau_\Mb(w_1) \keqsu \tau_\Mb(w_2)$.
\end{proof}

Recall that $\ell \in \nat$ was defined so that every $P \in \Pb$ belongs to $\su_\ell(A_s)$, and that $k = d(\ell+1)$.

\begin{fct} \label{fct:eqmu}
	We have $\mu_d(w_1) \eqsu{\ell+1} \mu_d(w_2)$.
\end{fct}

\begin{proof}
Let $\mu_d(w_1) = (b_1)\cdots (b_m)$ and $\mu_d(w_2) = (c_1)\cdots (c_n)$. By definition, $b_i$ and $c_j$ are length $d$ words in $A^*$ for all $1\le i < m$ and $1\le j < n$, and $b_m$ and $c_n$ are non empty words of length at most $d$.

Let $x \in A_d^*$ be a suffix of $\mu_d(w_1)$ of length $h \le \ell+1$. Then $x = (b_{m-h+1})\cdots (b_m)$. Note that the word $b_{m-h+1}\cdots b_m$ in $A^*$ has length $q \le hd \le k$, and that $x$ is its $\mu_d$-image. Note also that each suffix of $w_1$ is uniquely determined by the suffix of $\tau_\Mb(w_1)$ of the same length. 

Since $\tau_\Mb(w_1) \keqsu \tau_\Mb(w_2)$, the length $q$ suffix of $\tau_\Mb(w_1)$ is also a suffix of $\tau_\Mb(w_2)$, and hence the suffix of length $q$ of $w_1$, namely $b_{m-h+1}\cdots b_m$, is also a suffix of $w_2$. In addition, the lengths of $w_1$ and $w_2$ are congruent modulo $d$ by Fact~\ref{fct:congtau}, so we have $b_{m-i} = c_{n-i}$ for every $0\le i < h$. It follows that $x$ is also a suffix of $\mu_d(w_2)$.

This proves the announced fact by symmetry.
\end{proof}

We now conclude the proof of Lemma~\ref{lem:thecut}. Since every $P \in \Pb$ belongs to $\su_\ell(A_d)$ and since $\mu_d(w_1) \eqsu{\ell+1} \mu_d(w_2)$ (Fact~\ref{fct:eqmu}), we have $\ppart{\mu_d(w_1)} = \ppart{\mu_d(w_2)}$.

Since $|w_1|$ and $|w_2|$ are congruent modulo $d$ (Fact~\ref{fct:congtau}), we know that $v_1 = \varepsilon$ if and only if $v_2 = \varepsilon$. If $v_1 = v_2 = \varepsilon$, then $w_1 = u_1$ and $w_2 = u_2$ and we already showed that $\ppart{\mu_d(w_1)} = \ppart{\mu_d(w_2)}$. If instead $v_1 \neq \varepsilon$ and $v_2 \neq \varepsilon$, then Fact~\ref{fct:cuts} shows that $\mu_d(w_1) = \mu_d(u_1) \cdot (v_1)$ and $\mu_d(w_2) = \mu_d(u_2) \cdot (v_2)$. Using the fact that $\mu_d(w_1) \eqsu{\ell+1} \mu_d(w_2)$, we find that $v_1 = v_2$ and $\mu_d(u_1) \eqsu{\ell} \mu_d(u_2)$. It follows that $\ppart{\mu_d(u_1)} = \ppart{\mu_d(u_2)}$ since every $P \in \Pb$ belong to $\su_\ell(A_d)$.

\subsection{Proof of Lemma~\ref{lem:themorph}}\label{proof of 5.13}

We define a morphism $\alpha\colon (\Ub \times (\Mb \times A))^* \to (\Pb \times A_d)^*$ as follows. Let $(U,(M,a))$ be a letter in $\Ub \times (\Mb \times A)$. By definition of \Mb, $M = M_i$ for some $i < d$.

Suppose first that $i= d - 1$ and $U$ contains an element of $\tau_\Mb(A^*)$. By Lemma~\ref{lem:thecut}, there exist $P\in\Pb$ and $v\in A_{d-1}$ such that, for every word $w \in A^*$ such that $\tau_\Mb(w) \in U$, the $d$-cut of $w$ is of the form $(u,v)$ with $\ppart{u} = P$. Note also that $(va)$ is a letter of $A_d$ since $|v| \leq d-1$. We then let
\[
\alpha((U,(M,a))) = (P,(va)).
\]
Otherwise (that is: if $i \neq d-1$ or $U \cap \tau_\Mb(A^*) = \emptyset$), we let $\alpha((U,(M,a))) = \varepsilon$. Note in particular that $\alpha$ is not length increasing.

Let us now verify that $\alpha$ satisfies the desired property: we show that if $w \in A^*$ has $d$-cut $(u,v)$, then
%
$$\tau_\Pb(\mu_d(u))  = \alpha(\tau_\Ub(\tau_\Mb(w)).$$
%
We proceed by induction on the length of $w$. If $w = \varepsilon$, it is immediate that $\tau_\Pb(\mu_d(u)) =\alpha(\tau_\Ub(\tau_\Mb(w))) = \varepsilon$. Now assume that $|w| \geq 1$. There are two cases depending on whether $|w| = 0 \mod d$ (i.e. $w = u$ and $v = \varepsilon$) or not.
	
If $w \neq 0 \mod d$, then $v = v'a$ for some $v' \in A^*$ and $a \in A$. Note that $w = uv = uv'a$. Applying Fact~\ref{fct:tagdecomp} to $\tau_\Mb$ and again to $\tau_\Ub$, we get
\begin{align*}
\tau_\Mb(w) &= \tau_{\Mb}(uv') \ (\wpart{uv'}{\Mb},a) \\
\tau_\Ub(\tau_\Mb(w)) &= \tau_\Ub(\tau_\Mb(uv')) \ (\wpart{\tau_{\Mb}(uv')}{\Ub},(\wpart{uv'}{\Mb},a)).
\end{align*}
By definition, $|u| =  0 \mod d$ and $|v| < d$, so $|v'| < d-1$ and $|uv'| \neq d-1 \mod d$ and hence $\wpart{uv'}{\Mb} \neq M_{d-1}$. In view of the definition of $\alpha$, it follows that $\alpha((\wpart{\tau_{\Mb}(uv')}{\Ub},(\wpart{uv'}{\Mb},a))) = \varepsilon$ and therefore
\[
\alpha(\tau_\Ub(\tau_\Mb(w))) = \alpha(\tau_\Ub(\tau_\Mb(uv'))).
\]
We now use the induction hypothesis on $uv'$, whose length is $|w|-1$ and whose $d$-cut is $(u,v')$: this yields the equalities $\alpha(\tau_\Ub(\tau_\Mb(uv'))) = \tau_\Pb(\mu_d(u))$ and hence $\tau_\Pb(\mu_d(u)) =\alpha(\tau_\Ub(\tau_\Mb(w)))$ as desired.

Suppose finally that $|w| = 0 \mod d$, so $w = u$ and $v = \varepsilon$. Let $w' \in A^*$ and $a \in A$ such that $w = w'a$. In particular, $|w'| = d-1 \mod d$ and hence $\wpart{w'}{\Mb} = M_{d-1}$. Two successive applications of Fact~\ref{fct:tagdecomp} show that 
\begin{align*}
\tau_\Mb(w) &= \tau_\Mb(w') \ (\wpart{w'}{\Mb},a) = \tau_\Mb(w') \ (M_{d-1},a),\\
\tau_\Ub(\tau_\Mb(w)) &=  \tau_\Ub(\tau_\Mb(w')) \  (\wpart{\tau_\Mb(w')}{\Ub},(M_{d-1},a)).
\end{align*}
Let $(u',v')$ be the $d$-cut of $w'$. By induction, we get $\tau_\Pb(\mu_d(u')) =\alpha(\tau_\Ub(\tau_\Mb(w')))$ and hence
\[
\alpha(\tau_\Ub(\tau_\Mb(w))) =  \tau_\Pb(\mu_d(u')) \cdot \alpha((\wpart{\tau_\Mb(w')}{\Ub},(M_{d-1},a))).
\]
By definition of $\alpha$, we have $\alpha((\wpart{\tau_\Mb(w')}{\Ub},(M_{d-1},a))) = (\ppart{\mu_d(u')},(v'a))$ and hence
\begin{align*}
\alpha(\tau_\Ub(\tau_\Mb(w))) &= \tau_\Pb(\mu_d(u')) \cdot (\ppart{\mu_d(u')},(v'a)) \\
&= \tau_\Pb(\mu_d(u') \cdot (v'a))\textrm{ by Fact~\ref{fct:tagdecomp}.}
\end{align*}
Finally since $w = w'a = u'v'a$ and $|u'| = 0 \mod d$, Fact~\ref{fct:blockcut} yields $\mu_d(w) = \mu_d(u') \ \mu_d(v'a)$. Moreover, since $|v'a| = d$, $\mu_d(v'a)$ is the single letter $(v'a) \in A_d$. Thus $\mu_d(w) = \mu_d(u') \ (v'a)$ and we conclude that $\alpha(\tau_\Ub(\tau_\Mb(w)))=  \tau_\Pb(\mu_d(w))$. Since $u = w$, this is the desired result.

\section{Transfer theorem}
\label{sec:transfer}

The main theorem of this section, Theorem~\ref{thm:transfer} below, states that, for any class \Cs closed under inverse length increasing morphisms, \benr{\Cs}-covering reduces to \Cs-covering. As explained at the beginning of Section~\ref{sec: main theorem}, this is the last ingredient for the proof of Theorem~\ref{thm:main}.

Let us start with an outline of the steps involved in this reduction. Since an input pair $(L,\Lb)$ for the \benr{\Cs}-covering problem in $A^*$ involves only finitely many regular languages in $A^*$, we can compute a monoid morphism $\eta\colon A^* \to M$ which recognizes every one of them. Once this morphism is fixed, we construct (Section~\ref{sec: sfwords}) a new alphabet \sfa called the \emph{alphabet of stably-formed words} and a map $L \mapsto \sfl{L}$ which associates a language \sfl{L} over \sfa with every language $L \subseteq A^*$ recognized by $\eta$.

Theorem~\ref{thm:transfer} (in Section~\ref{sec: transfer thm}) then states that for any pair $(L,\Lb)$ such that $L$ and the languages in \Lb are recognized by $\eta$, the announced reduction is given by the following equivalence:
\[
\text{$(L,\Lb)$ is \benr{\Cs}-coverable} \quad  \text{if and only if} \quad \text{$(\sfl{L},\sfl{\Lb})$ is \Cs-coverable}.
\]
The proof of Theorem~\ref{thm:transfer} is given in Sections~\ref{sec: half1} and~\ref{sec: half2}.

\subsection{Stably-formed words}\label{sec: sfwords}

Let $\eta\colon A^* \to M$ be a morphism into a finite monoid $M$. The \emph{alphabet of stably-formed words} associated to $\eta$ is the finite set $\sfa = \eta(A^+)$. To resolve ambiguity, an element $x\in \eta(A^+)$ will be written $(x)$ if we view it as a letter in \sfa. We denote by $\eval$ the natural \emph{evaluation} morphism $\eval\colon (\sfa)^* \to M$, given by $\eval((x))  = x$ for every letter $(x) \in \sfa$ (that is: for every $x\in \eta(A^+)$).

We will be interested in only certain words of $\sfa^*$, which we call \emph{stably-formed}. To define this class of words, we remind the reader of the notion of \emph{stability index} associated to the morphism $\eta$.

\begin{fct} \label{fct:stabind}
	If $\eta\colon A^* \to M$ is a morphism into a finite monoid $M$, there exists an integer $s \geq 1$ such that $\eta(A^s) = \eta(A^{2s})$. 
\end{fct}

\begin{proof}
Every element $x$ of a finite monoid $S$ has a positive idempotent power, and we let $n_x$ be the least such power. Now let $\omega(S) = \lcm\{n_x\mid x \in S\}$: then every $\omega(S)$-power is an idempotent of $S$. The powerset $2^M$ is a finite monoid for the following product operation: given $X_1,X_2 \in 2^M$,  $X_1 \cdot X_2 = \{x_1x_2 \mid w_1 \in X_1 \text{ and } x_2 \in X_2\}$. Letting $s = \omega(2^M)$ establishes Fact~\ref{fct:stabind}.
\end{proof}

We call \emph{stability index} of $\eta$ the least natural number $s \geq 1$ such that $\eta(A^s) = \eta(A^{2s})$. By definition of $s$, $\eta(A^s)$ is a sub-semigroup of $M$ which we call the \emph{stable semigroup} of $\eta$. Note that the stable semigroup of $\eta$ can also be viewed as a subset of the alphabet \sfa.

Now a word $w \in \sfa^*$ is said to be \emph{stably-formed} if either $w = \varepsilon$, or $w = (x_1)(x_2)\cdots (x_n)$ with $n \geq 1$ and the letters $(x_1), \ldots, (x_{n-1})$ belong to the stable semigroup of $\eta$ (the rightmost letter $(x_n)$ can be any letter in \sfa). The following fact is immediate from the definition.

\begin{fct} \label{fct:itisreg}
	For any morphism $\eta\colon A^* \to M$, the language of stably-formed words in $(\sfa)^*$ is regular.
\end{fct}

We now associate with every language $L \subseteq A^*$ recognized by $\eta$ a language $\sfl{L} \subseteq (\sfa)^*$ called the \emph{language of stably-formed words associated to $L$}:
$$\sfl{L} = \{w \in (\sfa)^* \mid \text{$w$ is stably-formed and $\eval(w) \in \eta(L)$}\}.$$
Note that $\sfl{L}$ is the intersection of $\eval\inv(\eta(L))$ and the set of stably-formed words, both regular languages: it follows that $\sfl{L}$ is regular as well.

Let us emphasize the fact that \sfl{L} is only defined when the language $L$ is recognized by $\eta$. We extend our definition to multisets of languages recognized by $\eta$: given such a multiset \Lb, we let $\sfl{\Lb} = \{\sfl{L} \mid L \in \Lb\}$, a multiset as well.

\subsection{The transfer theorem}\label{sec: transfer thm}

We now state the main theorem of this section, namely the reduction between the \benr{\Cs}-covering and the \Cs-covering problems.

\begin{theorem}\label{thm:transfer}
Let \Cs be a class closed under inverse length increasing morphisms, let $\eta\colon A^* \to M$ be a morphism into a finite monoid $M$ and let $s$ be the stability index of $\eta$. For every language $L_1$ and finite multiset of languages $\Lb_2$, all recognized by $\eta$, the following properties are equivalent:
\begin{enumerate}
	\item $(L_1,\Lb_2)$ is \benr{\Cs}-coverable;
	\item $(L_1,\Lb_2)$ is \benrp{\Cs}{s}-coverable;
	\item $(\sfl{L_1},\sfl{\Lb_2})$ is \Cs-coverable.
\end{enumerate}
\end{theorem}

Note that one may adapt this statement to handle the weaker separation problem (i.e. the special case of inputs $(L_1,\Lb_2)$ where $\Lb_2$ is a singleton by Fact~\ref{fct:septocove}).

\begin{corollary}\label{cor:septransfer}
	Let \Cs be a class closed under inverse length increasing morphisms, let $\eta\colon A^* \to M$ be a morphism into a finite monoid $M$ and let $s$ be the stability index of $\eta$. $s \geq 1$. For all languages $L_1, L_2$, both recognized by $\eta$, the following properties are equivalent:
	\begin{enumerate}
		\item $L_1$ is \benr{\Cs}-separable from $L_2$;
		\item $L_1$ is \benrp{\Cs}{s}-separable from $L_2$;
		\item \sfl{L_1} is \Cs-separable from \sfl{L_2}.
	\end{enumerate}
\end{corollary}

%
%
%

The remainder of this section is devoted to the proof of Theorem~\ref{thm:transfer}. \Cs, $\eta$ and $s$ are fixed throughout the proof, satisfying the hypotheses of the theorem. With reference to the statement of Theorem~\ref{thm:transfer}, we prove the implications $(1) \Rightarrow (3) \Rightarrow (2) \Rightarrow (1)$.

The implication $(2) \Rightarrow (1)$ is actually trivial since $\benrp{\Cs}{s} \subseteq \benr{\Cs}$ by definition. The implication $(1) \Rightarrow (3)$ is proved in Section~\ref{sec: half1}, and the implication $(3) \Rightarrow (2)$ is proved in Section~\ref{sec: half2}.

\subsection{\texorpdfstring{From \benr{\Cs}-covering to \Cs-covering}{From [C]-covering to C-covering}}
\label{sec: half1}

Our argument is based on the following technical proposition.

\begin{proposition} \label{prop:firstdir}
	Let $d \geq 1$. There exists a morphism $\alpha\colon (\sfa)^* \to A^*$ such that:
	\begin{enumerate}[label=(\roman*)]
		\item if $L \subseteq A^*$ is recognized by $\eta$ and $w \in (\sfa)^*$ is stably-formed, then $w \in \sfl{L}$ if and only if $\alpha(w) \in L$;
		\item for every $K \in \benrd{\Cs}(A)$, there exists $H_K \in \Cs(\sfa)$ such that, for any stably-formed word $w \in (\sfa)^*$, $w  \in H_K$ if and only if $\alpha(w) \in K$.
	\end{enumerate}
\end{proposition}

Before we establish Proposition~\ref{prop:firstdir}, let us prove the implication $(1) \Rightarrow (3)$ in Theorem~\ref{thm:transfer}. Let $L_1$ be a language and $\Lb_2$ be a finite multiset of languages, all recognized by $\eta$, such that $(L_1,\Lb_2)$ is \benr{\Cs}-coverable. We prove that $(\sfl{L_1},\sfl{\Lb_2})$ is \Cs-coverable.

Let \Kb be a \benr{\Cs}-cover of $L_1$ which is separating for $\Lb_2$. Since every $K \in \Kb$ belongs to \benr{\Cs}, Corollary~\ref{cor:commonindex} yields a natural number $d \geq 1$ such that $K \in \benrd{\Cs}$ for all $K \in \Kb$. Let $\alpha\colon (\sfa)^* \to A^*$ be the morphism given by Proposition~\ref{prop:firstdir} for this value of $d$ and let $\Hb = \{H_K \mid K \in \Kb\}$. It suffices to show that \Hb is a \Cs-cover of \sfl{L_1} which is separating for \sfl{\Lb_2}.

Let $w \in \sfl{L_1}$. By definition, $w$ is stably-formed and by Proposition~\ref{prop:firstdir}~(i), we have $\alpha(w) \in L_1$. Since \Kb is a cover of $L_1$, $\alpha(w) \in K$ for some $K \in \Kb$, and Proposition~\ref{prop:firstdir}~(ii) now shows that $w\in H_K$. Thus \Hb is a cover of $L_1$, and a \Cs-cover since every $H_K$ is in \Cs (Proposition~\ref{prop:firstdir}~(ii) again).

Let us now verify that \Hb is separating for \sfl{\Lb_2}. Let $H \in \Hb$: by definition, there exists $K\in \Kb$ such that $H =H_K$. Since \Kb is separating for $\Lb_2$, there exists $L_2 \in \Lb_2$ such that $K \cap L_2 = \emptyset$. We prove that $H \cap \sfl{L_2} = \emptyset$ by contradiction: if $w \in H \cap \sfl{L_2}$, then Proposition~\ref{prop:firstdir}~(i) yields $\alpha(w) \in L_2$ and Proposition~\ref{prop:firstdir}~(ii) yields $\alpha(w) \in K$. Thus $\alpha(w) \in K \cap L_2$, which is a contradiction. This concludes the proof of the implication $(1) \Rightarrow (3)$ in Theorem~\ref{thm:transfer}.

\proofof{Proposition~\ref{prop:firstdir}}
Let $d \geq 1$ and let $s \geq 1$ be the stability index of $\eta$. We define a morphism $\alpha\colon(\sfa)^*\to A^*$ as follows. Let $x \in \sfa = \eta(A^+)$. If $x \not\in \eta(A^s)$, we pick any non-empty word $w \in A^+$ such that $\eta(w) = x$ and we let $\alpha(x) = w$. If instead $x \in \eta(A^s)$, then we also have $x \in \eta(A^{ds})$ by definition of the stability index: we pick a word $w \in A^{ds}$ of length $ds$ such that $\eta(w) = x$ and we let $\alpha(x) = w$.

Let $L$ be a language recognized by $\eta$ and let $w = (x_1)\cdots (x_n) \in (\sfa)^*$ be a stably-formed word. We first want to show Proposition~\ref{prop:firstdir}~(1), namely that $w \in \sfl{L}$ if and only if $\alpha(w) \in L$. We have $\alpha(w) = \alpha(x_1) \cdots \alpha(x_n)$ and, by definition of $\alpha$, we have $\eta(\alpha(w)) = x_1 \cdots x_n = \eval(w)$. Since $L$ is recognized by $\eta$ and hence $\alpha(w) \in L$ if and only if $\eta(\alpha(w)) \in \eta(L)$, it follows that $\alpha(w) \in L$ if and only if $\eval(w) \in \eta(L)$. By definition of $\sfl{L}$, this is equivalent to $w\in \sfl{L}$, as announced.

To establish Proposition~\ref{prop:firstdir}~(2), we consider the morphism $\gamma\colon (\sfa)^* \to A_d^*$ defined by letting $\gamma(x) = \mu_d(\alpha(x))$ for each letter $x\in \sfa$. Since each $\alpha(x)$ is non-empty, the morphism $\gamma$ is length increasing.

\begin{fct}\label{fct: createmorph}
For every stably-formed word $w \in \sfa^*$, we have $\gamma(w) = \mu_d(\alpha(w))$.
\end{fct}

\begin{proof}
The proof is by induction on the length of $w$. If $w = \varepsilon$, then $\gamma(w)= \mu_d(\alpha(w)) = \varepsilon$. We now assume that $|w| \geq 1$ and we let $(x) \in \sfa$ be the leftmost letter in $w$: $w = (x) \ w'$ for some $w' \in (\sfa)^*$. Then we have
\begin{align*}
\gamma(w) &= \gamma(x) \ \gamma(w') \\
& = \mu_d(\alpha(x)) \ \gamma(w') \textrm{ by definition of $\gamma$}\\
& = \mu_d(\alpha(x)) \ \mu_d(\alpha(w')) \textrm{ by induction}.
\end{align*}
If $w' = \varepsilon$, then $w = x$ and we get $\gamma(x) = \mu_d(\alpha(x))$ as desired. If $w' \ne \varepsilon$, then $(x)$ is not the rightmost letter of $w$ and by definition of stably-formed words, we have $x \in \eta(A^s)$ (i.e. the stable semigroup). By definition of $\alpha$, it follows that $\alpha(x) \in A^+$ has length $sd$ and we can use Fact~\ref{fct:blockcut} to show that
\[
\gamma(w) = \mu_d(\alpha(x)) \ \mu_d(\alpha(w')) = \mu_d(\alpha((x) \ w')) = \mu_d(\alpha(w)),
\]
as desired.
\end{proof}

We want to show that, for any language $K \in \benrd{\Cs}(A)$, there exists $H_K \in \Cs(\sfa)$ such that for any stably-formed word $w \in (\sfa)^*$, $w  \in H_K$ if and only if $\alpha(w) \in K$.

By definition of $\benrd{\Cs}(A)$, there exists $P \in \Cs(A_d)$ such that $K = \mu_d\inv(P)$ and we let $H_K = \gamma\inv(P)$. Since \Cs is closed under inverse length increasing morphisms, we have $H_K \in \Cs(\sfa)$.

Now let $w\in \sfa^*$ be stably-formed. By definition of $H_K$, $w\in H_K$ if and only if $\gamma(w) \in P$. By Fact~\ref{fct: createmorph}, this is equivalent to $\mu_d(\alpha(w)) \in P$ and hence to $\alpha(w) \in \mu_d\inv(P) = K$, as desired.
\eopo

\subsection{\texorpdfstring{From \Cs-covering to \benr{\Cs}-covering}{From C-covering to [C]-covering}}
\label{sec: half2}

We turn to the implication $(3) \Rightarrow (2)$ in Theorem~\ref{thm:transfer}. We first introduce auxiliary maps $\sigma$ and $\rho$.

The morphism $\sigma\colon A_s^* \to \sfa^*$ is defined as follows: if $x$ is a letter in $A_s$ (i.e., $x\in A^+$ is a non-empty word of length at most $s$), we let $\sigma(x) = \eta(x) \in \sfa$ (recall that $\sfa = \eta(A^+)$). In particular, $\sigma$ is length increasing. We now define the map $\rho\colon A^* \to (\sfa)^*$ to be the composition $\rho = \sigma \circ \mu_s$. Note that $\rho$ is not a morphism. Useful properties of $\rho$ are summarized in the following lemma.

\begin{lemma}\label{lm:seconddir}
The following properties hold.
\begin{enumerate}[label=(\roman*)]
\item If $w\in A^*$, then $\rho(w)\in\sfa^*$ is stably-formed.

\item If $L\subseteq A^*$ is recognized by $\eta$, then $L = \rho\inv(\sfl{L})$.

\item If $K \in \Cs(\sfa)$, then $\rho\inv(K) \in \benrp{\Cs}{s}(A)$.
\end{enumerate}
\end{lemma}

\begin{proof}
Let $w\in A^*$. If $w = \varepsilon$, then $\rho(w) = \varepsilon$, which is stably-formed. If $w \ne \varepsilon$, then $\mu_s(w) = (x_1) \cdots (x_n)$ with $n\ge 1$ and $x_1,\ldots,x_{n-1} \in A^s$. Then $\rho(w) = \sigma(\mu_s(w)) = (\eta(x_1)) \cdots (\eta(x_n))$ and $\eta(x_1), \ldots,\eta(x_{n-1}) \in \eta(A^s)$: this is exactly the definition of a stably-formed word. Thus (i) holds.

Now let $L\subseteq A^*$ be recognized by $\eta$ and let $w\in A^*$. By definition of $\sfl{L}$, $L$ contains the empty word if and only if $\sfl{L}$ does. Now let $w \in A^+$. Then $\mu_s(w) = (x_1) \cdots (x_n)$ with $n\ge 1$ and each $x_i$ is a non-empty word of length at most $s$. In particular $w = x_1 \cdots x_n$ and $\rho(w) = (\eta(x_1)) \cdots (\eta(x_n))$. Then we have $\eval(\rho(w)) = \eta(x_1)\cdots \eta(x_n) = \eta(w)$. By definition of $\sfl{L}$ and since $\rho(w)$ is stably-formed, we have $\rho(w) \in \sfl{L}$ if and only if $\eval(\rho(w)) \in \eta(L)$, if and only if $\eta(w) \in \eta(L)$. This is equivalent to $w\in L$ since $\eta$ recognizes $L$ and hence, (ii) holds.

Finally, let $K \in \Cs(\sfa)$ and let $P = \sigma\inv(K) \subseteq A_s^*$. In particular $P \in \Cs(A_s)$ since \Cs is closed under inverse length increasing morphisms. Moreover $\rho\inv(K) = \mu_s\inv(\sigma\inv(K)) = \mu_s\inv(P)$, which belongs to$\benrp{\Cs}{s}(A)$ by definition. This concludes the proof of (iii).
\end{proof}

We now prove the implication $(3) \Rightarrow (2)$ in Theorem~\ref{thm:transfer}. Let $L_1$ be a language and $\Lb_2$ be a finite multiset of languages, all recognized by $\eta$, such that $(\sfl{L_1},\sfl{\Lb_2})$ is \Cs-coverable. Let \Kb be a \Cs-cover of $\sfl{L_1}$ which is separating for $\sfl{\Lb_2}$ and let
\[
\Ub = \{\rho\inv(K) \mid K \in \Kb\}.
\]
In order to complete the proof, we verify that \Ub is a \benrp{\Cs}{s}-cover of $L_1$ which is separating for $\Lb_2$.

Lemma~\ref{lm:seconddir}~(iii) shows that each element of \Ub is in $\benrp{\Cs}{s}(A)$. Now let $w\in L_1$. By Lemma~\ref{lm:seconddir}~(ii), $\rho(w) \in \sfl{L_1}$ and since \Kb is a cover of \sfl{L_1}, $\rho(w) \in K$ for some $K \in \Kb$. Then $w \in \rho\inv(K)$, which is an element of \Ub by definition. Thus \Ub is a \benrp{\Cs}{s}-cover of $L_1$.

Now let $U \in \Ub$, say, $U = \rho\inv(K)$ for some $K \in \Kb$. Since \Kb is separating for $\sfl{\Lb_2}$, there exists $L_2 \in \Lb_2$ such that $K \cap \sfl{L_2} = \emptyset$. Then $\rho\inv(K) \cap \rho\inv(\sfl{L_2}) = \emptyset$. Now $U = \rho\inv(K)$ and $\rho\inv(\sfl{L_2}) = L_2$ by Lemma~\ref{lm:seconddir}~(ii). Therefore $U \cap L_2 = \emptyset$, which shows that \Ub is separating for $\Lb_2$.

\section{Conclusion}
\label{sec:conc}

We showed that for every \pvarie \Cs, the covering problem for $(\Cs \circ \su) \circ \md$ is effectively reducible to the same problem for \Cs. Exploiting the connection between language theoretic enrichment and logical enrichment by local and modular predicates, we used this result to obtain that covering is decidable for the fragments \fowm, \fodwsm, \siwsm{n} for $n =1,2,3$ and \bswsm{n} for $n =1,2$.

Naturally, a downside of this result is that we are only able to hand \md-enrichment for classes which have been built with \su-enrichment. Therefore a natural question is whether there exists a complementary theorem which states that for any \pvarie \Cs, covering for $\Cs \circ \md$ reduces to the same problem for \Cs. Such a theorem would make it possible to handle logical classes equipped with the modular predicates but not the local ones, such as \fodwm, \siwm{n} or \bswm{n}. Let us point out that this would be a complementary result and not a generalization of our main theorem: the classes of the form $\Cs \circ \su$, which we can handle, are \plivaris but usually not \pvaries.

Finally let us point out that a natural generalization of our results concerns languages of infinite words, for which regularity and modular predicates are well-defined. Such a generalization, for \su-enrichment and the addition of local predicates, is treated in detail by Place and Zeitoun in~\cite{2018:PlaceZeitoun}. An analogous reasoning for \md-enrichment and the addition of modular predicates would yield analogous results.



\bibliographystyle{alpha}
\bibliography{PRW}

\end{document}